\documentclass[11pt,reqno,a4paper]{article}

\usepackage{amsmath}
\usepackage{amsthm}
\usepackage{amsfonts}
\usepackage{amssymb}
\usepackage[dvips]{graphicx}
\usepackage[dvips]{color}
\usepackage{latexsym}
\usepackage{enumerate}
\usepackage{xspace}

\newtheorem{thm}{Theorem}[section]
\newtheorem{cor}[thm]{Corollary}
\newtheorem{lemma}[thm]{Lemma}

\newtheorem{example}[thm]{Example}

\numberwithin{equation}{section}

\def \bI {\Bbb I}
\def \bN {\Bbb N}

\def \bR {\Bbb R}
\def \bK {\Bbb K}

\def \bC {\Bbb C}

\def \cA {{\cal A}}
\def \cB {{\cal B}}

\def \cH {{\cal H}}
\def \cI {{\cal I}}

\def \cM {{\cal M}}
\def \cS {{\cal S}}

\def \cX {{\cal X}}

\def \cE {{\cal E}}

\def \and {\, \mbox{\rm and}\, }
\def \sinc {\,{\rm sinc}\,}
\def \dist {\,{\rm dist}\,}

\def \span {\,{\rm span}\,}
\def \supp {\,{\rm supp}\,}

\def \tr {\,{\rm tr}\,}

\DeclareMathOperator*{\argmin}{argmin}

\begin{document}

\title{Optimal Sampling Points in Reproducing Kernel Hilbert Spaces\thanks{Supported
by Guangdong Provincial Government of China through the
``Computational Science Innovative Research Team" program.}}

\author{\quad Rui Wang\thanks{School of Mathematics, Jilin University, Changchun 130012, P. R. China. E-mail address: {\it rwang11@jlu.edu.cn}. Supported by Natural Science
Foundation of China under grants 11071250 and 11126149.} \quad and \quad Haizhang Zhang\thanks{Corresponding author. School of Mathematics and Computational Science and Guangdong Province Key Laboratory of Computational Science,
Sun Yat-sen University, Guangzhou 510275, P. R. China. E-mail address: {\it zhhaizh2@sysu.edu.cn}. Supported in part by Natural Science Foundation of China under grants 11101438 and 91130009, by the US Army Research Office, and by SRF for ROCS, SEM.}}

\date{}
\maketitle
\begin{abstract}
The recent developments of basis pursuit and compressed sensing seek to extract information from as few samples as possible. In such applications, since the number of samples is restricted, one should deploy the sampling points wisely. We are motivated to study the optimal distribution of finite sampling points. Formulation under the framework of optimal reconstruction yields a minimization problem. In the discrete case, we estimate the distance between the optimal subspace resulting from a general Karhunen-Lo\`{e}ve transform and the kernel space to obtain another algorithm that is computationally favorable. Numerical experiments are then presented to illustrate the performance of the algorithms for the searching of optimal sampling points.

\smallskip

\noindent {\bf Keywords:} sampling points, optimal distribution, reproducing kernels, the Karhunen-Lo\`{e}ve transform
\end{abstract}

\section{Introduction}
\setcounter{equation}{0}

Functions describing natural phenomenon or social activities need to be converted into discrete data that can be handled by modern computers. From this viewpoint, sampling is the foundation for signal processing and communication. The subject origined from the celebrated Shannon sampling theorem \cite{Shannon}, which gurantees the complete reconstruction of a band-limited function from its values on some equally-spaced points.
The elegant result motivates many follow-up studies, making sampling an important research subject in applied mathematics. We shall give a brief and partial introduction to the history and progresses.

Mathematically, sampling means to evaluate a function. To ensure the stability, it is arguable that sampling should only take place in function spaces where point evaluations are continuous. Such spaces when endowed with an inner product structure arise in many other areas of mathematics. They are termed as the reproducing kernel Hilbert spaces (RKHS), as by the Riesz's lemma there exists a function that is able to reproduce the function values through the inner product. In Shannon's theorem, the space of functions that are band-limited to $[-\pi,\pi]$ and are equipped with the inner product of $L^2(\bR)$ is an RKHS with the sinc function as its reproducing kernel. This interpretation gives the hope of searching for Shannon-type complete reconstruction formula for other RKHS. It was found in \cite{Nashed} that as long as one has a frame or a Riesz basis formed by the reproducing kernel, then a Shannon-type sampling formula is immediately available by the general theory of frames. They showed that many past sampling formulae can be obtained in this manner. Recently, the approach has been generalized to reproducing kernel Banach spaces \cite{Zhang2009} by frames for Banach spaces via semi-inner-products, \cite{Zhang2011}.

Shannon type formulae enable us to have lossless representation of a function that is usually defined on an uncountable continuous domain using countable data. Going from uncountable to countable is a remarkable progress. However, countable is still infinite and computers can not store or handle infinitely many data. This raises the question of how to reconstruct a function from its finite sample. For the crucial band-limited functions, two modified Shannon series have been proposed in the literature \cite{Jagerman,Qian}, where it was shown that over-sampling can lead to exponentially decaying approximation error. Sampling data often comes with some cost. When it is available, one is inclined to use as accurate reconstruction methods as possible. It has long been known that in the maximum sense, the best way of reconstruction in an RKHS is via the minimal norm interpolation \cite{MicchelliRivlin}. The approximation error for over-sampling in the Paley-Wiener space of band-limited functions is estimated in \cite{MXZ}.

In this note, we focus on another important question in sampling, which is seldom considered in the literature. Usually the number of sampling points in a practical application is limited. When that number is fixed, we ask what is the best strategy of deploying the sampling points, under the condition that the best reconstruction method is engaged. The study is also motivated by the recent development in basis pursuit \cite{CDS} and compressed sensing \cite{CRT}, which seek to extract information from as few samples as possible. Since the number of samples is restricted, we should of course distribute the sampling points wisely.

We shall formulate the question in the next section. It will become clear that the solution of the problem amounts to approximating the subspace spanned by the first few eigenvectors of a compact operator. When the operator is of finite rank, the eigenvectors can be obtained by the well-known Karhunen-Lo\`{e}ve transform ( also called principal component analysis in engineering). To extend the algorithm to operators usually defined by integrals in this application, we shall establish a general Karhunen-Lo\`{e}ve transform in Section 3. An alternative approach by subspace approximation that can significantly reduce computational cost will be introduced in Section 4. Various examples by numerical experiments will be presented in Section 5. The study will lead to algorithms for the searching of the optimal distribution of finite sampling points for commonly-used RKHS.

\section{Formulation}
\setcounter{equation}{0}
A natural choice of background function spaces for sampling is reproducing kernel Hilbert spaces (RKHS). Let $X$ be a prescribed metric space where functions of interest are defined. An RKHS on $X$ is a Hilbert space $\cH$ of functions on $X$ such that for each $x\in X$, the point evaluation functional
$$
\delta_x(f):=f(x),\ \ f\in\cH
$$
is continuous. An RKHS $\cH$ possesses a unique {\it reproducing kernel} \cite{Aronszajn1950}, which is a function on $X\times X$ characterized by the properties that for all $f\in \cH$ and $x\in X$, $K(x,\cdot)\in \cH$ and
\begin{equation}\label{reproducing}
f(x)=(f,K(x,\cdot))_{\cH},
\end{equation}
where $(\cdot,\cdot)_{\cH}$ denotes the inner product on $\cH$. On the other hand, the reproducing kernel $K$ uniquely determines the RKHS $\cH$. Thus, the RKHS of a reproducing kernel $K$ is usually denoted by $\cH_K$. For more information on reproducing kernels, see \cite{Aronszajn1950,LMWX,MXZ2006,Schoenberg}.

We emphasize that an RKHS should first be a Hilbert space of functions, which implies that a function in the space has zero norm if and only if it vanishes everywhere. For instance, the Paley-Wiener space
$$
\cB:=\{f\in C(\bR^d)\cap L^2(\bR^d):\supp\hat{f}\subseteq [-\pi,\pi]^{d}\}
$$
is an RKHS. In this paper, the Fourier transform $\hat{f}$ of $f\in L^1(\bR^d)$ is defined by
$$
\hat{f}(\xi):=\frac1{(\sqrt{2\pi})^{2d}}\int_{\bR^d}f(x)e^{-i x\cdot \xi}dx,\ \ \xi\in\bR^d,
$$
 where $x\cdot \xi$ is the standard inner product on $\bR^d$. The norm on $\cB$ inherits from that in $L^2(\bR^d)$. The reproducing kernel for the Paley-Wiener space $\cB$ is the sinc function
$$
\sinc (x,y)=\prod_{j=1}^d\frac{\sin\pi(x_j-y_j)}{\pi(x_j-y_j)},\ \ x,y\in\bR^d.
$$

We consider the deployment of finite sampling points in an RKHS in this paper. Let $\cH_K$ be an RKHS on a metric space $X$ and the number $n$ of sampling points be fixed. The choice of the sampling points depends on the method of reconstruction and the measurement of the approximation error. For most applications, one desires to reconstruct values of the function considered on a compact subspace $\Omega\subseteq X$. The reconstruction error will be measured by the norm in $L^p_\mu(\Omega)$. Here $p\in[1,+\infty]$, $\mu$ is a finite positive Borel measure on $\Omega$, and the Banach space $L^p_\mu(\Omega)$ consists of Borel measurable functions $f$ on $\Omega$ that satisfy
$$
\|f\|_{L^p_\mu(\Omega)}:=\biggl(\int_\Omega |f(t)|^pd\mu(t)\biggr)^{1/p}<+\infty,\ \ 1\le p<+\infty
$$
and
$$
\|f\|_{L^\infty_\mu(\Omega)}:=\inf \{c\ge0: |f|\le c\mbox{ almost everywhere on } \Omega\mbox{ with respect to }\mu\}<+\infty.
$$

We shall assume throughout this paper that $K$ is continuous on $\Omega$. We observe by the reproducing property (\ref{reproducing}) for all $x,y\in X$ and $f\in\cH_K$ that
$$
\begin{array}{ll}
|f(x)-f(y)|&=|(f,K(x,\cdot)-K(y,\cdot))_{\cH_K}|\\&\le \|f\|_{\cH_K}\|K(x,\cdot)-K(y,\cdot)\|_{\cH_K}\\
&=\|f\|_{\cH_K}\sqrt{K(x,x)-K(x,y)-K(y,x)+K(y,y)}.
\end{array}
$$
Therefore, every function $f\in \cH_K$ belongs to the space $C(\Omega)$ of continuous functions on $\Omega$ equipped with the usual maximum norm. Consequently, $\cH_K\subseteq L^p_\mu(\Omega)$ for all $1\le p\le +\infty$ and all finite Borel measures $\mu$ on $\Omega$.

Let $\cX:=\{x_j:1\le j\le n\}$ be a choice of $n$ sampling points. The sample data of a function $f\in\cH_K$ is hence of the form
$$
\cI_\cX(f):=\{f(x_j):1\le j\le n\}.
$$
A reconstruction method $\cA$ is then a mapping from $\bC^\cX$ to $L^p_\mu(\Omega)$. For a particular $f\in\cH_K$, the reconstruction error is measured by
$$
\|f-\cA(\cI_\cX(f))\|_{L^p_\mu(\Omega)}.
$$
We then follow the general setting of optimal sampling in \cite{MicchelliRivlin} and \cite{ZhangThesis}, that is, we measure the performance of a reconstruction method $\cA$ by
$$
\rho(\cA,\cX):=\sup\{\|f-\cA(\cI_\cX(f))\|_{L^p_\mu(\Omega)}:\ f\in\cH_K,\ \|f\|_{\cH_K}\le 1\}.
$$
Since we are concerned with the optimal choice of sampling points only, we shall try to remove the reconstruction method from the picture. To this end, we shall use the optimal reconstruction algorithm $\cA_\cX$ for each choice of sampling points $\cX$. Namely,
\begin{equation}\label{cacx}
\rho(\cA_\cX,\cX)=\inf\{\rho(\cA,\cX):\mbox{ among all mapping }\cA\mbox{ from }\bC^\cX\mbox{ to }L^p(\Omega,d\mu)\}.
\end{equation}
Finally, our problem reduces to finding the sampling points $\cX$ that minimizes the function
$$
\cE(\cX):=\rho(\cA_\cX,\cX),\ \ \cX\in X^n.
$$

The optimal reconstruction algorithm $\cA$ is known to be the minimal norm interpolation \cite{MicchelliRivlin,ZhangThesis}. The following lemma also gives the reconstruction error.

\begin{lemma}\label{optimalalgorithm}
For each set of sampling points $\cX\in X^n$, the optimal reconstruction method $\cA_\cX$ satisfying (\ref{cacx}) is given by
$$
\cA_\cX(\cI_\cX(f)):=\argmin\{\|g\|_{\cH_K}:\ g\in\cH_K,\ \cI_\cX(g)=\cI_\cX(f)\}.
$$
The associated reconstruction error is of the form
$$
\rho(\cA_\cX,\cX):=\sup\{\|f\|_{L^p_\mu(\Omega)}:\ f\in \cH_K,\ \|f\|_{\cH_K}\le 1,\ \cI_\cX(f)=0\}.
$$
\end{lemma}

A reproducing kernel determines everything about the corresponding RKHS. The following simple observation fulfills this hope. Set
\begin{equation}\label{kernelspace}
\cS_\cX:=\span\{K(t,\cdot):\ t\in \cX\}.
\end{equation}
We shall impose another assumption through the paper that for every set of pairwise distinct sampling points $\cX$, the matrix
$$
K[\cX]:=[K(x_j,x_k):\ 1\le j,k\le n]
$$
is nonsingular. A reproducing kernel is at the same time a positive-definite function, \cite{Aronszajn1950}. Thus, $K[\cX ]$ is strictly positive-definite. With this assumption, $\cS_\cX$ is $n$-dimensional with the orthonormal basis
\begin{equation}\label{onb}
u_j=\sum_{k=1}^n\alpha_{jk}K(x_k,\cdot),\ \ 1\le j\le n,
\end{equation}
where
\begin{equation}\label{coealpha}
[\alpha_{j,k}:1\le j,k\le n]=(K[\cX])^{-1/2}.
\end{equation}

\begin{cor}
Let $\phi_\cX$ be defined by
$$
\phi_\cX(x):=\dist(K(x,\cdot),\cS_\cX):=\min\{\|K(x,\cdot)-g\|_{\cH_K}:\ g\in\cS_\cX\},\ \ x\in X.
$$
Then it holds true for each $x\in\Omega$ that
\begin{equation}\label{reformulation}
\phi_\cX(x)=\sup\{|f(x)|:f\in\cH_K,\ \|f\|_{\cH_K}\le 1,\ \cI_\cX(f)=0\}.
\end{equation}
Furthermore, for each $p\ge 1$
\begin{equation}\label{erro}
\rho(\cA_\cX,\cX)\le\|\phi_\cX\|_{L^p_\mu(\Omega)},
\end{equation}
and for the special case when $p=+\infty$,
\begin{equation}\label{errorinfty}
\rho(\cA_\cX,\cX)=\|\phi_\cX\|_{L^\infty_\mu(\Omega)}.
\end{equation}
\end{cor}
\begin{proof}
Let $f$ be an arbitrary function in $\cH_K$ such that $\|f\|_{\cH_K}\le1$ and $\cI_\cX(f)=0$. Then by the reproducing property (\ref{erro}),
$$
(f,K(x_j,\cdot))_{\cH_K}=0\mbox{ for all }1\le j\le n.
$$
It follows that $f$ is orthogonal to every $g\in\cS_\cX$. We hence see that
$$
|f(x)|=|(f,K(x,\cdot)-g)_{\cH_K}|\le \|f\|_{\cH_K}\|K(x,\cdot)-g\|_{\cH_K}\le \|K(x,\cdot)-g\|_{\cH_K}.
$$
As the above equation is true for all $g\in \cS_\cX$, we get that $|f(x)|\le \phi_\cX(x)$, $x\in X$. As a result, it holds for all $p\ge 1$ that
$$
\rho(\cA_\cX,\cX)\le \|\phi_\cX\|_{L^p_\mu(\Omega)}.
$$

On the other hand, letting $f$ be the orthogonal projection of $K(x,\cdot)$ onto $\cS_\cX$ and then be normalized to a unit vector yields (\ref{reformulation}). Thus, for $p=+\infty$,
$$
\begin{array}{rl}
\rho(\cA_\cX,\cX)&=\sup\{\sup\{|f(x)|:x\in\Omega\}:f\in\cH_K,\ \|f\|_{\cH_K}\le 1,\ \cI_\cX(f)=0\}\\
&=\sup\{\sup\{|f(x)|:f\in\cH_K,\ \|f\|_{\cH_K}\le 1,\ \cI_\cX(f)=0\}:x\in\Omega\}\\
&=\|\phi_\cX\|_{L^\infty_\mu(\Omega)},
\end{array}
$$
which proves (\ref{errorinfty}).
\end{proof}

By the above corollary, we shall hence try to minimize the quantity $\|\phi_\cX\|_{L^p_\mu(\Omega)}$ as a way to bound the intrinsic error $\rho(\cA_\cX,\cX)$. A simple calculation tells that
\begin{equation}\label{phi}
\phi_\cX^2(x)=K(x,x)-\sum_{j=1}^n|u_j(x)|^2,\ \ x\in X.
\end{equation}
This together with (\ref{onb}) and (\ref{coealpha}) gives a function about $\cX$ that needs to be minimized. The complicated form of the function coped with the nonlinearity of the reproducing kernel makes directly minimizing this function rather difficult. Before discussing alternative computational methods, we present two simple examples to demonstrate that the optimal points might not be equally-spaced distributed in the reconstruction domain $\Omega$.

\begin{example} In this trivial example, we let $X=\bR^d$, $\Omega$ a compact subset in $\bR^d$ and $n=1$. The reproducing kernel is given by a radial basis function
$$
K(x,y):=\varphi(\|x-y\|),\ \ x,y\in\bR^d
$$
where $\|x\|$ denotes the standard Euclidean norm on $\bR^d$. The function $\varphi$ is a univariate function that defines a reproducing kernel in the above manner. By Schoenberg's theorem \cite{Schoenberg}, $\varphi(\sqrt{\cdot})$ must be a completely monotone function. In particular, $\phi$ is nonincreasing. For simplicity, we also assume that $\varphi(0)=1$. We shall use the space $C(\Omega)$ to measure the reconstruction error. The optimal sampling point $x_0$ is hence the minimizer of
which leads to
$$
\begin{array}{ll}
\min_{t\in\bR^d}\|\phi_t^2\|_{C(\Omega)}^2&\displaystyle{=\min_{t\in \mathbb{R}^d}\max_{x\in\Omega} 1-|K(t,x)|^2}\\
&\displaystyle{=1-\max_{t\in\mathbb{R}^d}\min_{x\in\Omega} |K(t,x)|^2}\\
&\displaystyle{=1-\max_{t\in\mathbb{R}^d}\min_{x\in\Omega} \varphi^2(\|x-t\|)}\\
&\displaystyle{=1-\varphi^2(\min_{t\in\mathbb{R}^d}\max_{x\in\Omega}\|x-t\|)}.
\end{array}
$$
By the above equation, $x_0$ is the point has a minimal radius $r$ for which $\Omega\subseteq \{x:\|x-x_0\|\le r\}$. Particularly, for $d=1$,
we should choose $x_0$ as the mid-point of the end points of $\Omega$.
\end{example}

Unlike the above example, our second example shows that nonlinearity could occur as the number of sampling points exceeds $1$. The analysis of this example of two sampling points is already rather tedious but elementary, and is thus omitted.

\begin{example}\label{example2}
In this example, we let $X=\bR$, $\Omega=[a,b]\subseteq \bR$, $n=2$ and consider the exponential kernel
$$
K(x,y):=e^{-\|x-y\|}, \ x,y\in\bR.
$$
In this case, for $\cX:=\{x_1,x_2\}$,
$$
\phi_\cX(x)=1-V(x,x_1,x_2)
$$
where
$$
V(x,x_1,x_2):=\frac{e^{-2\|x_1-x\|}+e^{-2\|x_2-x\|}-2e^{-(\|x_1-x\|+\|x_2-x\|+\|x_1-x_2\|)}}{1-e^{-2\|x_1-x_2\|}}.
$$
The optimal sampling points $x_1,x_2$ is the minimizer of
$$
\sup_{x_1,x_2\in\mathbb{R}}\min_{x\in\Omega}V(x,x_1,x_2).
$$
Let $L:=b-a$. After some careful but elementary analysis, it can be found that the optimal sampling points are
\begin{equation}\label{optpoint1}
x_1=a-\frac{1}{2}\ln\left(\frac{-e^{-L}+\sqrt{e^{-2L}+8e^{-L}}}{2}\right)
\end{equation}
and
\begin{equation}\label{optpoint2}
x_2=b+\frac{1}{2}\ln\left(\frac{-e^{-L}+\sqrt{e^{-2L}+8e^{-L}}}{2}\right).
\end{equation}
\end{example}

Although measuring the reconstruction error by the maximum norm in $C(\Omega)$ seems the most natural and the maximum norm dominates other $L^p$ norms, finding the extrema of a multivariate function is always difficult. A Hilbert space norm can often save computation efforts. From this consideration, we restrict ourself to the choice $L^2_\mu(\Omega)$ in the rest of the paper. In the case when $\Omega:=\{y_k:1\le k\le m\}\subseteq X$ with $m\gg n$ and $\mu(\{y_k\})=1/m$ for $1\le k\le m$, the $n$-dimensional subspace $\cS_0$ that minimizes
$$
\inf\left\{\frac1m\sum_{k=1}^m\dist^2(K(y_k,\cdot),\cS):\ \cS\mbox{ is an }n\mbox{-dimensional subspace of }\cH_K\right\}
$$
is given by the Karhunen-Lo\`{e}ve transform. More specifically, $\cS_0$ is spanned by the eigenfunctions corresponding to the largest $n$ eigenvalues of the compact positive bounded linear operator $T$ on $\cH_K$ given by
$$
T(f):=\frac1m\sum_{k=1}^m f(y_k)K(y_k,\cdot).
$$
The process of computing the eigenfunctions and eigenvalues of this operator is also known as kernel principal component analysis in machine learning \cite{Scholkopf2001a}.
Of course, the story is not over yet as the space we are looking for should be of the form (\ref{kernelspace}).
Our idea is to find sampling points $\cX$ for which $\cS_\cX$ best approximates the subspace spanned by the first $n$ eigenfunctions of $T$. Before we estimate the distance between these two subspaces of $\cH_K$, we first show that for general reconstruction error, the minimization problem
\begin{equation}\label{KLproblem}
\min\left\{\int_\Omega\dist^2(K(x,\cdot),\cS)d\mu(x):\ \cS\mbox{ is an }n\mbox{-dimensional subspace of }\cH_K\right\}
\end{equation}
can still be reduced to computing the first $n$ eigenfunctions of a compact positive bounded linear operator on $\cH_K$. We shall prove such a Karhunen-Lo\`{e}ve transform exists for general measure $\mu$.

\section{A general Karhunen-Lo\`{e}ve transform}
\setcounter{equation}{0}

The purpose of this section is to show that the subspace that minimizes (\ref{KLproblem}) is spanned by the first $n$ eigenfunctions of a compact positive bounded linear operator. We shall prove this result under a very general setting.

Let $\cH$ be an infinite-dimensional separable Hilbert space, $(\Omega,\cM,\mu)$ be a measure space, that is, $\cM$ is a $\sigma$-algebra consisting of certain subsets of $\Omega$ and $\mu$ is a finite positive measure on $\cM$. We assume that there is a function $F:\Omega\to \cH$ such that for each $u\in\cH$, the function
$$
\omega\to (F(\omega),u)_{\cH}
$$
is measurable with respect to $\cM$ and such that $\|F(\cdot)\|_{\cH}\in L^2_\mu(\Omega,\cM)$. For a fixed $n\in\bN$, we want to find an $n$-dimensional subspace $V$ of $\cH$ that approximates $F(\Omega)$ well. By measuring the approximation of each candidate subspace $V$ as
$$
\cE(V):=\int_\Omega \dist^2(F(\omega),V)d\mu(\omega),
$$
the optimal approximating subspace $\cS_n$ is the one that minimizes the above error among all $n$-dimensional subspaces of $\cH$. A Karhunen-Lo\`{e}ve transform for this general question is presented below.

\begin{thm}\label{KL}
The operator $T:\cH\to\cH$ determined by
\begin{equation}\label{KLoperator}
(Tu,v)_{\cH}=\int_\Omega (u,F(\omega))_\cH(F(\omega),v)_\cH d\mu(\omega),\ u,v\in \cH,
\end{equation}
is compact positive bounded linear. The optimal $n$-dimensional subspace $\cS_n$ that satisfies
$$
\cE(\cS_n)=\inf\{\cE(V):\ V\mbox{ is an }n\mbox{-dimensional subspace of }\cH\}
$$
is given by $\cS_n=\span\{e_j:1\le j\le n\}$, where $e_j$'s are the orthonormal eigenfunctions corresponding to the largest $n$ eigenvalues of $T$.
\end{thm}
\begin{proof}
Let $v\in\cH$ be fixed. Then for each $u\in\cH$, we observe that
$$
\begin{array}{ll}
\displaystyle{\left|\int_\Omega (u,F(\omega))_\cH(F(\omega),v)_\cH d\mu(\omega)\right|}&\displaystyle{\le
\int_\Omega\left|(u,F(\omega))_\cH\right|\left|(F(\omega),v)_\cH\right|d\mu(\omega)}\\
&\displaystyle{\le
 C_F\,\|u\|_\cH\|v\|_\cH},
\end{array}
$$
where
$$
C_F:=\int_\Omega \|F(\omega)\|_\cH^2d\mu(\omega).
$$
It implies that
$$
u\to \int_\Omega (u,F(\omega))_\cH(F(\omega),v)_\cH d\mu(\omega)
$$
is a bounded linear functional on $\cH$. By the Riesz representation theorem, there exists a unique vector $w_v$ associated with $v$ such that
$$
\int_\Omega (u,F(\omega))_\cH(F(\omega),v)_\cH d\mu(\omega)=(u,w_v)_\cH.
$$
We denote the mapping sending $v$ to $w_v$ by $T$. It is clear that this operator is linear. Moreover, we have
$$
|(u,Tv)_\cH|=\left|\int_\Omega (u,F(\omega))_\cH(F(\omega),v)_\cH d\mu(\omega)\right|\le C_F\|u\|_\cH\|v\|_\cH.
$$
Therefore, $\|Tv\|_\cH\le C_F\|v\|_\cH$, implying that $T$ is bounded. We also see that for all $u\in\cH$
$$
(Tu,u)_\cH=\int_\Omega |(u,F(\omega))_\cH|^2 d\mu(\omega)\ge0.
$$
Thus, $T$ is positive.

We next show that $T$ is compact. To this end, let $u_j$ be a bounded sequence in $\cH$. Then $Tu_j$ is bounded as well. As $\cH$ is reflexive, its unit ball is weakly compact. We may hence assume that $Tu_j$ converges weakly to some $u_0$ in $\cH$. In other words,
$$
\lim_{j\to\infty} (Tu_j,v)_\cH=(u_0,v)_\cH\mbox{ for all }v\in\cH.
$$
We shall prove that $Tu_j$ converges to $u_0$ strongly in $\cH$. Note that
$$
(Tu_j-u_0,Tu_j-u_0)_\cH=(Tu_j-u_0,Tu_j)_\cH-(Tu_j-u_0,u_0)_\cH.
$$
As $(Tu_j-u_0,u_0)_\cH\to0$ as $j\to\infty$, it suffices to show that
$$
\lim_{j\to\infty}(Tu_j-u_0,Tu_j)_\cH=0.
$$
We observe from the definition of $T$ that
$$
(Tu_j-u_0,Tu_j)_\cH=\int_\Omega ((Tu_j,F(\omega))_\cH-(u_0,F(\omega))_\cH)(F(\omega),u_j)_\cH d\mu(\omega).
$$
For each $\omega\in\Omega$, $(Tu_j,F(\omega))_\cH\to (u_0,F(\omega))_\cH$ as $Tu_j$ converges weakly to $u_0$. As a result, there holds
\begin{eqnarray*}
&&\lim_{j\to\infty}|((Tu_j,F(\omega))_\cH-(u_0,F(\omega))_\cH)(F(\omega),u_j)_\cH|\\
&\le& \|F(\omega)\|_\cH\sup_{j}\|u_j\|_\cH\lim_{j\to\infty}|(Tu_j,F(\omega))_\cH-(u_0,F(\omega))_\cH|=0.
\end{eqnarray*}
Furthermore,
$$
|((Tu_j,F(\cdot))_\cH-(u_0,F(\omega))_\cH)(F(\cdot),u_j)_\cH|\le (\|T\|\|u_j\|_\cH^2+\|u_0\|_\cH\|u_j\|_\cH)\|F(\cdot)\|_\cH^2\in L^1_\mu(\Omega,\cM).
$$
The above equations together imply by the Lebesgue dominated convergence theorem that
$$
\lim_{j\to\infty}(Tu_j-u_0,Tu_j)_\cH=\lim_{j\to\infty}\int_\Omega ((Tu_j,F(\omega))_\cH-(u_0,F(\omega))_\cH)(F(\omega),u_j)_\cH d\mu(\omega)=0.
$$
Therefore, $\|Tu_j-u_0\|_\cH\to0$ as $j\to\infty$. We have hence proved that $T$ is a positive compact bounded linear operator on $\cH$.

Turning to the last claim of the theorem, we let $V$ be an $n$-dimensional subspace of $\cH$ with the orthonormal basis $f_j$, $1\le j\le n$. Then
$$
\cE(V)=\int_\Omega \|F(\omega)\|_\cH^2-\sum_{j=1}^n |(F(\omega),f_j)_\cH|^2d\mu(\omega)=\int_\Omega\|F(\omega)\|_\cH^2d\mu(\omega)-\sum_{j=1}^n(Tf_j,f_j)_\cH.
$$
Thus, the question amounts to finding an orthonormal sequence $\{f_j:1\le j\le n\}$ in $\cH$ that maximizes the sum
$$
\sum_{j=1}^n (Tf_j,f_j)_\cH.
$$
The analysis of this last part is the same as that for the standard Karhunen-Lo\`{e}ve transform, that is, the optimal sequence is achieved by the orthonormal eigenfunctions corresponding to the largest $n$ eigenvalues of $T$.
\end{proof}

Returning to the sampling, we specify $\Omega$ to be a compact subset of the input space $X$, $\mu$ to be a finite positive Borel measure on $X$, $K$ to be a continuous kernel on $X$, and
$$
F(t):=K(t,\cdot),\ \ t\in\Omega.
$$
By Theorem \ref{KL}, the bounded linear operator $T$ from $\cH_K$ to $\cH_K$ determined by
\begin{equation*}\label{operatorTK1}
(Tf,g)_{\cH_K}=\int_\Omega f(t)\overline{g(t)}d\mu(t),\ \ f,g\in\cH_K
\end{equation*}
is positive and compact. It is of the explicit form
\begin{equation}\label{operatorTK1}
(Tf)(x)=\int_\Omega f(t)K(t,x)d\mu(t),\ \ x\in X,\ f\in\cH_K.
\end{equation}
For each $n$-dimensional subspace $V$ of $\cH_K$ with the orthonormal basis $\{u_j:1\le j\le n\}$,
$$
\cE(V)=\int_\Omega\dist^2(K(t,\cdot),V)d\mu(t)=\int_\Omega K(t,t)d\mu(t)-\sum_{j=1}^n (Tu_j,u_j)_{\cH_K}.
$$
An orthonormal basis for $\cS_\cX=\span\{K(x_j,\cdot):1\le j\le n\}$ is given by (\ref{onb}). Thus,
$$
\cE(V)=\int_\Omega K(t,t)d\mu(t)-\sum_{j=1}^n\sum_{k=1}^n\sum_{l=1}^n\alpha_{jk}\alpha_{lj}(T(K(x_k,\cdot)),K(x_l,\cdot))_{\cH_K}.
$$
Setting
$$
\bK_{k,l}:=(T(K(x_k,\cdot)),K(x_l,\cdot))_{\cH_K}=\int_\Omega K(x_k,t)K(t,x_l)d\mu(t),\ \ 1\le k,l\le n,
$$
we conclude that the optimal sampling set $\cX$ is the solution of
\begin{equation}\label{algorithm1}
\max_{\cX\in X^n}\sum_{j,k,l=1}^n\alpha_{jk}\alpha_{lj}\bK_{k,l}=\max_{\cX\in X^n}\mbox{tr}\left((K[\cX])^{-1/2}\bK (K[\cX])^{-1/2}\right)=\max_{\cX\in X^n}\mbox{tr}\left(\bK^{1/2} (K[\cX])^{-1}\bK^{1/2}\right),
\end{equation}
where $\tr(M)$ stands for the trace of a square matrix $M$. When the eigenfunctions and eigenvalues of the operator $T$ is known, one has a different formulation of the above optimization problem. Let $e_i$, $i\in\bI$ be all the orthonormal eigenfunctions of $T$ with a positive eigenvalue $\lambda_i$. We see for all $1\le k,l\le n$ that
\begin{eqnarray*}
\bK_{k,l}=(T(K(x_k,\cdot)),K(x_l,\cdot))_{\cH_K}&=&\sum_{i,i'\in\bI}(K(x_k,\cdot),e_i)_{\cH_K}(e_{i'},K(x_l,\cdot))_{\cH_K}(Te_i,e_{i'})_{\cH_K}\\
&=&\sum_{i\in\bI}\lambda_ie_i(x_l)e_i(x_k).
\end{eqnarray*}

Practically, we are most concerned with the case when $\Omega$ has finite cardinality that is considerably larger than $n$. In this situation, $T$ has finite rank. Assume that $\bI=\{1,2,\ldots,m\}$ and set
$$
\Lambda_{il}:=\delta_{i,l}\sqrt{\lambda_i},\ D_{ki}:=e_i(x_k),\ \ 1\le i,l\le m,\ \ 1\le k\le n.
$$
With these notations, $\bK=(D\Lambda)(D\Lambda)^T$. When $T$ is of finite rank, this together with the fact that for a square matrix $A$, $\tr(AA^T)=\tr(A^TA)$ yields an equivalent formulation of (\ref{algorithm1})
\begin{equation}\label{algorithm2}
\max_{\cX\in X^n}\tr\left(\Lambda D^T(K[\cX])^{-1}D\Lambda\right).
\end{equation}
Computing all the eigenfunctions and eigenvalues of $T$ can be costly when $m$ is large. Instead of attacking (\ref{algorithm1}) or (\ref{algorithm2}) directly, we shall relax (\ref{algorithm2}) to use only the $n$ eigenfunctions of $T$ corresponding to the first $n$ largest eigenvalues of $T$, which can often be obtained efficiently by the standard Karhunen-Lo\`{e}ve algorithm. Following the idea described at the end of Section 2, we shall achieve this by estimating the distance between $\cS_\cX$ and the one spanned by the first $n$ eigenfunctions of $T$.

\section{Subspace approximation}
\setcounter{equation}{0}

We now let $e_j$, $1\le j\le n$ be the orthonormal eigenfunctions of $T$, defined as in (\ref{operatorTK1}), corresponding to the largest $n$ eigenvalues of $T$. We assume that these eigenvalues are positive. By Theorem \ref{KL}, the subspace $\cS_T:=\span\{e_i:1\le i\le n\}$ is a minimizer of optimization problem (\ref{KLproblem}). We wish to find sampling points $\cX$ such that $\cE(\cS_\cX)-\cE(\cS_T)$ is small, where for a closed subspace $V$ of $\cH_K$,
$$
\cE(V)=\int_\Omega \dist^2(K(t,\cdot),V)d\mu(t).
$$
To this end, we first observe that for any closed subspaces $U$ and $V$ of $\cH_K$, $|\cE(U)-\cE(V)|$ can be bounded by the subspace distance between $U$ and $V$.

Denote by $P_V$ the orthogonal projection operator from $\cH_K$ onto $V$. The distance between two closed subspaces $U$ and $V$ of $\cH_K$ is defined by
$$
\dist(U,V):=\|P_U-P_V\|,
$$
where $\|P_U-P_V\|$ is the operator norm of $P_U-P_V$, that is,
$$
\|P_U-P_V\|=\sup_{f\in\cH_K}\frac{\|P_U(f)-P_V(f)\|_{\cH_K}}{\|f\|_{\cH_K}}.
$$
Apparently, the above supremum can be restricted to the closed subspace spanned by the union of $U$ and $V$.

\begin{lemma}\label{boundbydist}
It holds for any two closed subspaces $U$ and $V$ of $\cH_K$ that
\begin{equation}\label{boundbydisteq}
|\cE(U)-\cE(V)|\le 2K_\Omega\dist(U,V),
\end{equation}
where
$$
K_\Omega:=\int_\Omega K(t,t)d\mu(t).
$$
\end{lemma}
\begin{proof}
Denote by $I$ the identity operator. We estimate that
\begin{eqnarray*}
&&\left|\dist^2(K(t,\cdot),U)-\dist^2(K(t,\cdot),V)\right|\\
&=&\left|\|(I-P_U)K(t,\cdot)\|_{\cH_K}^2-\|
(I-P_V)K(t,\cdot)\|_{\cH_K}^2\right|\\
&\le&\|(I-P_U)K(t,\cdot)-(I-P_V)K(t,\cdot)\|_{\cH_K}(\|(I-P_U)
K(t,\cdot)\|_{\cH_K}+\|(I-P_V)K(t,\cdot)\|_{\cH_K})\\
&\le& \|P_U-P_V\|\|K(t,\cdot)\|_{\cH_K}(\|K(t,\cdot)\|_{\cH_K}+\|K(t,\cdot)\|_{\cH_K})\\
&=&2K(t,t)\dist(U,V),
\end{eqnarray*}
from which (\ref{boundbydisteq}) follows.
\end{proof}

According to the above lemma, we face to figure out the distance between subspaces $\cS_\cX$ and $\cS_T$. To this end,
we introduce some notations. Set
$$
f_j:=\sum_{k=1}^n (K(x_j,\cdot),e_k)_{\mathcal{H}_K}e_k=\sum_{k=1}^n \overline{e_k(x_j)}e_k,\ \ 1\le j\le n.
$$
In other words, $f_j$ is the orthogonal projection of $K(x_j,\cdot)$ onto $\cS_T$. Also, set
$$
h_j:=K(x_j,\cdot)-f_j,\ \ 1\le j\le n.
$$
Accordingly, we define two positive definite matrices by letting
$$
\mathbf{A}:=[(f_k,f_j)_{\cH_K}:1\le j,k\le n]\ \mbox{and}\ \mathbf{B}:=[(h_k,h_j)_{\cH_K}:1\le j,k\le n]
$$
We shall assume that $\mathbf{A}$ and $\mathbf{B}$ are both nonsingular. It will be shown in the proof below that ${\bf A}+{\bf B}=K[\cX]^T$. We assume in this section that $K[\cX]$ is nonsingular as well.
\begin{lemma}\label{subspacedistlemma}
If the matrix $\mathbf{E}:=[e_k(x_j):1\le j,k\le n]$ is nonsingular then
\begin{equation}\label{subspacedisteq}
\dist(\cS_\cX,\cS_T)=\sqrt{1-\frac{1}{\lambda_{\mbox{max}}(K[\mathcal{X}]^T(\mathbf{E}\mathbf{E}^{*})^{-1})}}.
\end{equation}
where $\lambda_{\mbox{max}}(M)$ denotes the largest eigenvalue of a square matrix $M$. If $\mathbf{E}$ is singular then
$\dist(\cS_\cX,\cS_T)\geq1$.
\end{lemma}
\begin{proof}
If $\mathbf{E}$ is singular then there exists a nonzero function in $\cS_\cX$ that is orthogonal to $\cS_T$.
It follows immediately that $\dist(\cS_\cX,\cS_T)\geq1$.

Suppose that $\mathbf{E}$ is nonsingular. By the nonsingularity of $\mathbf{E}$, $\cS_T$ is identical with the following subspace of $\cH_K$:
$$
U:=\span \{f_j:\ 1\le j\le n\}.
$$
The space $\cS_\cX$ coincides with $\widetilde{U}:=\span\{f_j+h_j:1\le j\le n\}$. For later use, we also introduce another two subspaces of $\cH_K$:
$$
V:=\span \{h_j:1\le j\le n\}\ \mbox{and}\ W:=\span\{f_j,h_j:1\le j\le n\}.
$$
We first observe that
$$
\dist(U,\widetilde{U})=\sup_{w\in W}\frac{\|P_U(w)-P_{\widetilde{U}}(w)\|_{\cH_K}}{\|w\|_{\cH_K}}.
$$
Any $w\in W$ can be represented as $w=u+v$, where $u\in U$ and $v\in V$. By definition, we have $(f_j,h_k)_{\mathcal{H}_K}=0$ for all $1\le j,k\le n$,
which yields that $U$ is orthogonal to $V$. We get that
\begin{eqnarray}\label{dist}
{\rm dist}^2(U,\widetilde{U})&=&\sup_{u\in U, v\in V}\frac{\|(P_{U}-P_{\widetilde{U}})(u+v)\|^2_{\cH_K}}{\|u+v\|^2_{\cH_K}}\nonumber\\
&=&\sup_{u\in U, v\in V}\frac{\|u-P_{\widetilde{U}}(u+v)\|^2_{\cH_K}}{\|u\|^2_{\cH_K}+\|v\|^2_{\cH_K}}.
\end{eqnarray}
To estimate (\ref{dist}), we first give $P_{\widetilde{U}}(u+v)$ explicitly. To this end, we assume that
\begin{equation}\label{projection}
P_{\widetilde{U}}(u+v)=\sum_{k=1}^nc_k(f_k+h_k)
\end{equation}
for some $c_j\in\bC$. By the characterization of orthogonal projections, we get  the equations
$$
(u+v-\sum_{k=1}^nc_k(f_k+h_k), f_j+h_j)_{\cH_K}=0, \ \ {1\le j\le n},
$$
which leads to
\begin{equation}\label{cequation}
\sum_{k=1}^nc_k((f_k,f_j)_{\cH_K}+(h_k,h_j)_{\cH_K})=(u,f_j)_{\cH_K}+(v,h_j)_{\cH_K},\ \  {1\le j\le n}.
\end{equation}
Set $\mathbf{c}:=[c_k:{1\le k\le n}]^T$. For each
$$
u=\sum_{k=1}^nu_kf_k\in U \ \mbox{and} \ v=\sum_{k=1}^n v_kh_k\in V,
$$
we set $ \mathbf{u}:=[u_k:{1\le k\le n}]^T$ and $\mathbf{v}:=[v_k:{1\le k\le n}]^T$. Then equation (\ref{cequation}) can be rewritten in a matrix form
$$
(\mathbf{A}+\mathbf{B})\mathbf{c}=\mathbf{A}\mathbf{u}+\mathbf{B}\mathbf{v}.
$$
Thus we obtain
\begin{equation}\label{c}
\mathbf{c}=(\mathbf{A}+\mathbf{B})^{-1}(\mathbf{A}\mathbf{u}+\mathbf{B}{\bf v}).
\end{equation}
By (\ref{projection}),  we have
\begin{eqnarray*}
\|u-{P}_{\widetilde{U}}(u+v)\|^2_{\cH_K}&=&\|u-\sum_{k=1}^nc_k(f_k+h_k)\|^2_{\cH_K}\\
&=&\|u-\sum_{k=1}^nc_k f_k\|^2_{\cH_K}+\|\sum_{k=1}^nc_k h_k\|^2_{\cH_K}\\
&=&\|\sum_{k=1}^n(u_k-c_k)f_k\|^2_{\cH_K}+\|\sum_{k=1}^nc_k h_k\|^2_{\cH_K}\\
&=&(\mathbf{u}-\mathbf{c})^{*}\mathbf{A}(\mathbf{u}-\mathbf{c})+\mathbf{c}^{*}\mathbf{B}\mathbf{c}.
\end{eqnarray*}
Substituting (\ref{c}) into the above equation, we get that
\begin{eqnarray*}
\|u-{P}_{\widetilde{U}}(u+v)\|^2_{\cH_K}&=&\mathbf{u}^{*}\mathbf{A}\mathbf{u}-\mathbf{u}^{*}\mathbf{A}(\mathbf{A}+\mathbf{B})^{-1}(\mathbf{A}\mathbf{u}+\mathbf{B}\mathbf{v})\\
&&-(\mathbf{A}\mathbf{u}+\mathbf{B}\mathbf{v})^{*}(\mathbf{A}+\mathbf{B})^{-1}\mathbf{A}\mathbf{u}\\
&&+(\mathbf{A}\mathbf{u}+\mathbf{B}\mathbf{v})^{*}(\mathbf{A}+\mathbf{B})^{-1}(\mathbf{A}\mathbf{u}+\mathbf{B}\mathbf{v})\\
&=&\mathbf{u}^{*}\mathbf{A}\mathbf{u}-\mathbf{u}^{*}\mathbf{A}(\mathbf{A}+\mathbf{B})^{-1}\mathbf{A}\mathbf{u}+\mathbf{v}^{*}\mathbf{B}(\mathbf{A}+\mathbf{B})^{-1}\mathbf{B}\mathbf{v}.
\end{eqnarray*}
Together with the fact that
$$
\|u\|^2_{\cH_K}=\mathbf{u}^{*}\mathbf{A}\mathbf{u},\ \mbox \ \|v\|^2_{\cH_K}=\mathbf{v}^{*}\mathbf{B}\mathbf{v}
$$
the above equation leads to
\begin{equation}\label{dist1}
{\rm dist}^2(U,\widetilde{U})=\sup_{\mathbf{u}, \mathbf{v}\in\bC^n}\frac{\mathbf{u}^{*}\mathbf{A}\mathbf{u}-\mathbf{u}^{*}\mathbf{A}(\mathbf{A}+\mathbf{B})^{-1}
\mathbf{A}\mathbf{u}+\mathbf{v}^{*}\mathbf{B}(\mathbf{A}+\mathbf{B})^{-1}\mathbf{B}\mathbf{v}}{\mathbf{u}^{*}\mathbf{A}\mathbf{u}+\mathbf{v}^{*}\mathbf{B}\mathbf{v}}.
\end{equation}
Let $\mathbf{a}:=\mathbf{A}^{1/2}\mathbf{u}$ and $\mathbf{b}:=\mathbf{B}^{1/2}\mathbf{v}$. By introducing a matrix
$$
{\bf M}:=\left(\begin{array}{cc}
\mathbf{I}-\mathbf{A}^{1/2}(\mathbf{A}+\mathbf{B})^{-1}\mathbf{A}^{1/2}& 0\\
0&\mathbf{B}^{1/2}(\mathbf{A}+\mathbf{B})^{-1}\mathbf{B}^{1/2}
\end{array}
\right),
$$
we get that
\begin{equation}\label{estimate}
{\rm dist}^2(U,\widetilde{U})=\sup_{\mathbf{a},\mathbf{b}\in\mathbb{C}^n}\frac{\left[\begin{array}{cc}
\mathbf{a}\\
\mathbf{b}
\end{array}
\right]^*\mathbf{M}\left[\begin{array}{cc}
\mathbf{a}\\
\mathbf{b}
\end{array}
\right]}{\left\|\left[\begin{array}{cc}
\mathbf{a}\\
\mathbf{b}
\end{array}
\right]\right\|_2^2}=\|\mathbf{M}\|_2,
\end{equation}
where $\|\cdot\|_2$ denotes the standard Euclidean norm of a vector or the spectral norm of a square matrix. On the one hand, we have
\begin{eqnarray*}
\|\mathbf{I}-\mathbf{A}^{1/2}(\mathbf{A}+\mathbf{B})^{-1}\mathbf{A}^{1/2}\|_2
&=&\|\mathbf{A}^{-1/2}\mathbf{B}(\mathbf{A}+\mathbf{B})^{-1}\mathbf{A}^{1/2}\|_2\nonumber\\
&=&\lambda_{\mbox{max}}(\mathbf{A}^{-1/2}\mathbf{B}(\mathbf{A}+\mathbf{B})^{-1}\mathbf{A}^{1/2}).
\end{eqnarray*}
Since the matrix $\mathbf{A}$ is nonsingular, the matrix $\mathbf{A}^{-1/2}\mathbf{B}(\mathbf{A}+\mathbf{B})^{-1}\mathbf{A}^{1/2}$ has the same eigenvalues
with the matrix $\mathbf{B}(\mathbf{A}+\mathbf{B})^{-1}$. Hence, we have
\begin{equation}\label{estimate1}
\|\mathbf{I}-\mathbf{A}^{1/2}(\mathbf{A}+\mathbf{B})^{-1}\mathbf{A}^{1/2}\|_2=\lambda_{\mbox{max}}(\mathbf{B}(\mathbf{A}+\mathbf{B})^{-1}).
\end{equation}
On the other hand, by the nonsingularity of the matrix $\mathbf{B}$, we also have
\begin{eqnarray}\label{estimate2}
\|\mathbf{B}^{1/2}(\mathbf{A}+\mathbf{B})^{-1}\mathbf{B}^{1/2}\|_2&=&\lambda_{\mbox{max}}
(\mathbf{B}^{1/2}(\mathbf{A}+\mathbf{B})^{-1}\mathbf{B}^{1/2})\nonumber\\
&=&\lambda_{\mbox{max}}(\mathbf{B}(\mathbf{A}+\mathbf{B})^{-1}).
\end{eqnarray}
Combining (\ref{estimate1}) with (\ref{estimate2}), we get that
\begin{equation*}\label{estimate3}
{\rm dist}^2(U,\widetilde{U})=\lambda_{\mbox{max}}(\mathbf{B}(\mathbf{A}+\mathbf{B})^{-1}).
\end{equation*}
For each $1\le j,k\le n$, there holds
\begin{eqnarray*}
(h_k,h_j)_{\cH_K}&=&(K(x_k,\cdot)-f_k,K(x_j,\cdot)-f_j)_{\cH_K}\\
&=&K(x_k,x_j)-\overline{f_j(x_k)}-f_k(x_j)+(f_k,f_j)_{\cH_K}\\
&=&K(x_k,x_j)-\overline{f_j(x_k)},
\end{eqnarray*}
which leads to $\mathbf{B}=K[\mathcal{X}]^T-\mathbf{A}$. Hence, we obtain
$$
{\rm dist}^2(U,\widetilde{U})=\lambda_{\mbox{max}}((K[\mathcal{X}]^T-\mathbf{A})(K[\mathcal{X}]^T)^{-1})=1-
\rho_{\mbox{min}}(\mathbf{A}(K[\mathcal{X}]^T)^{-1})=1-\frac{1}{\lambda_{\mbox{max}}(K[\mathcal{X}]^T\mathbf{A}^{-1})}.
$$
It follows from $\mathbf{A}=\mathbf{E}\mathbf{E}^*$ that there holds (\ref{subspacedisteq}).
\end{proof}

Combining Lemmas \ref{boundbydist} and \ref{subspacedistlemma}, we obtain a bound for the distance between $\cS_\cX$ and the optimal subspace $\cS_T$ and give the last optimization problem for the searching of optimal sampling points.
\begin{thm}\label{subspacedist}
If $\mathbf{E}$ is nonsingular then
$$
\cE(\cS_\cX)-\cE(\cS_T)\le 2K_\Omega\sqrt{1-\frac{1}{\lambda_{\mbox{max}}(K[\mathcal{X}]^T(\mathbf{E}\mathbf{E}^{*})^{-1})}}.
$$
\end{thm}

We conclude that the subspace approximation approach leads to the following problem
\begin{equation}\label{algorithm3}
\min_{\cX\in X^n}\lambda_{\mbox{max}}(K[\mathcal{X}]^T(\mathbf{E}\mathbf{E}^{*})^{-1})
\end{equation}
to be solved for the searching of optimal sampling points. We remark that when the measure $\mu$ is discrete as in most practical applications, (\ref{algorithm3}) is computationally favorable over (\ref{algorithm1}). The reason is that in this case, an orthonormal basis for the optimal subspace $\cS_T$ can be easily computed by the Karhunen-Lo\`{e}ve transform. At each stage of searching for the candidate sampling points $\cX$, the matrix ${\bf E}$ can be obtained efficiently and the major computation occurs with taking the inverse of a matrix. As comparison, algorithm (\ref{algorithm1}) additional requires the computation of the matrix $\bK$ and its square root.

\section{Numerical Experiments}
\setcounter{equation}{0}In this section, we give some numerical experiments to illustrate the performance of algorithms (\ref{algorithm1}) and (\ref{algorithm3}) for the searching of optimal sampling points. To this end, we first recall by Lemma \ref{optimalalgorithm} that for an obtained $n$ sampling points $\cX=\{x_j:1\le j\le n\}\in X^n$, the optimal method of reconstructing $\tilde{f}$ of a given function $f\in\cH_K$ from the sampled data $f(\cX)$ is given by
\begin{equation}\label{reconstruction}
\tilde{f}(x)=\sum_{j=1}^n\alpha_jK(x_j,x),\quad x\in X,
\end{equation}
where the coefficients $\alpha_j,1\le j\le n,$ are the unique solution of the linear system
\begin{equation}\label{coefficients}
\sum_{j=1}^nK(x_j,x_k)\alpha_j=f(x_k),\quad 1\le k\le n.
\end{equation}
Here we assume throughout the section that the kernel matrix $K[\cX]$ is nonsingular.

Therefore, our procedure of experiments is as follows. We shall consider the Gaussian kernel
$$
K(x,y)=e^{-\|x-y\|^2},\ x,y\in\bR^d
$$
and the sinc kernel
$$
K(x,y)=\prod_{j=1}^d \frac{\sin\pi(x_j-y_j)}{\pi(x_j-y_j)}, \ x,y\in\bR^d.
$$
Let $K$ be one of these two kernels, $X=\Omega\in\bR^d$ be compact, and $\mu$ be a selected Borel measure on $\Omega$. We then solve the optimization problem (\ref{algorithm1}) or (\ref{algorithm3}) to obtain $n$ sampling points $\cX_{\mbox{opt}}$, which are to be compared with the commonly used equally-spaced sampling points $\cX_{\mbox{equ}}$. For this purpose, we randomly generate 100 finite linear combinations $f$ of the kernel
$$
f=\sum c_jK(z_j,\cdot)
$$
as the target functions to be sampled, where both the coefficients $c_j$'s and the locations $z_j$'s will be randomly generated by the uniform distribution. For each of those target functions $f$, we then compute by (\ref{reconstruction}) and (\ref{coefficients}) the reconstructed functions $\tilde{f}_{\mbox{opt}}$ and $\tilde{f}_{\mbox{equ}}$ from the sampled values of $f$ on $\cX_{\mbox{opt}}$ and $\cX_{\mbox{equ}}$, respectively. Finally, the relative approximation errors
$$
\cE_{\mbox{opt}}:=\frac{\|\tilde{f}_{\mbox{opt}}-f\|_{L^2(\Omega)}}{\|f\|_{L^2(\Omega)}}, \ \ \cE_{\mbox{equ}}:=\frac{\|\tilde{f}_{\mbox{equ}}-f\|_{L^2(\Omega)}}{\|f\|_{L^2(\Omega)}}.
$$
are calculated.

To present the results, we shall first plot $\cX_{\mbox{opt}}$ against $\cX_{\mbox{equ}}$. The mean and standard deviation of the difference $\cE_{\mbox{equ}}-\cE_{\mbox{opt}}$ for the 100 pairs of relative errors will then be tabulated. Finally, we plot the 100 pairs of relative errors for a visual comparison, followed by discussion.\newline

\noindent{{\bf  \large Experiment 1: algorithm (\ref{algorithm1}), $K=$ the one-dimensional Gaussian kernel, $n=12$, $\Omega=[-3,3]$, $\mu=$ the Lebesgue measure on $\Omega$.}}

{\bf Figure 5.1} Distribution of the obtained 12 optimal sampling points (marked with a star) and the equally-spaced points (marked with a circle) on $\Omega=[-3,3]$.
\begin{center}
\scalebox{0.5}[0.6]{\includegraphics*{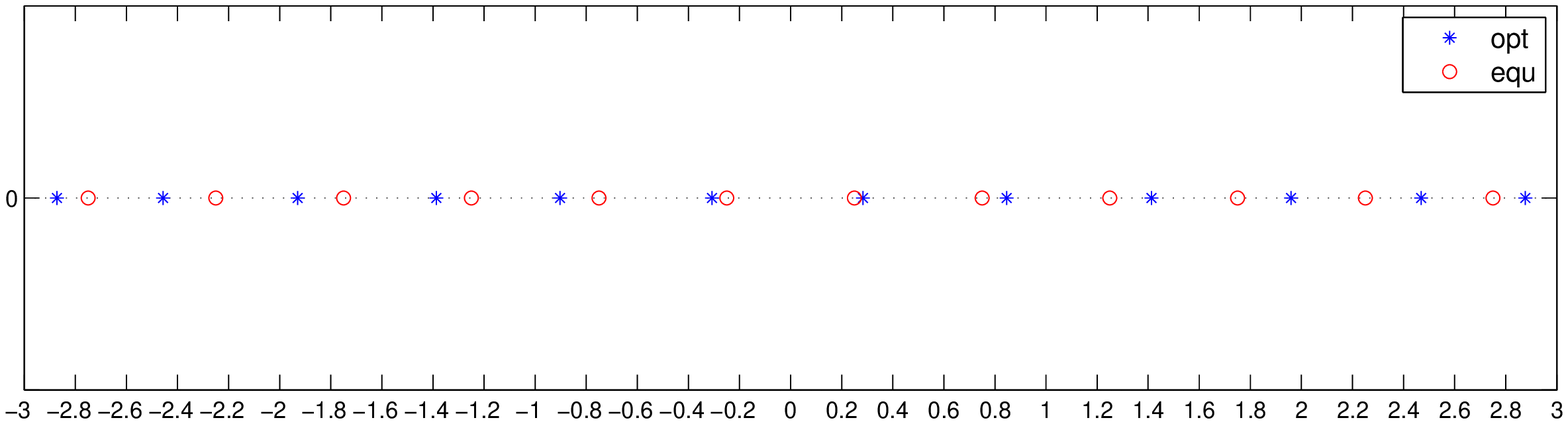}}
\end{center}

{\bf Table 5.1} The mean and standard deviation of the improvement $\cE_{\mbox{equ}}-\cE_{\mbox{opt}}$.
$$
\begin{array}{cc}
\hline\hline
\mbox{mean}&\mbox{standrad deviation}\\
0.3705\times 10^{-3}  &  0.5049\times 10^{-3}\\
\hline\hline
\end{array}
$$

{\bf Figure 5.2} Relative approximation errors $\cE_{\mbox{opt}}$ (marked with a circle) and $\cE_{\mbox{equ}}$ (marked with a star).
\begin{center}
\scalebox{0.5}[0.6]{\includegraphics*{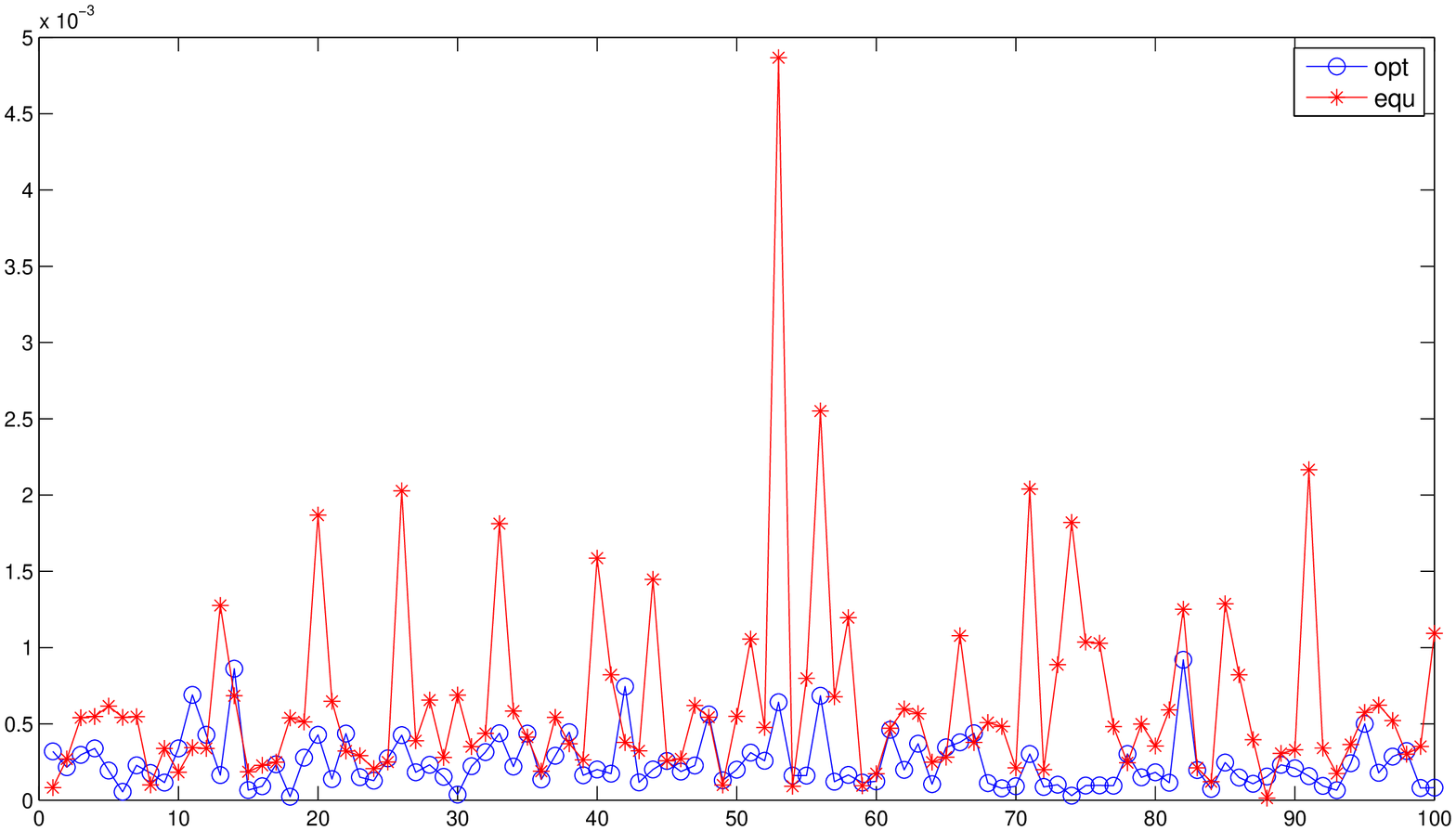}}
\end{center}

We observe that for the 100 pairs of relative approximation errors, there are only 20 pairs for which $\cE_{\mbox{opt}}$ is larger than $\cE_{\mbox{equ}}$. Recall that the optimal sampling points are designed to ensure that it is best in average for all the functions in the RKHS $\cH_K$. Therefore, situations where the optimal sampling points perform worse than the equally-spaced sampling points could indeed occur. For this experiment, one sees that in those 20 instances, the relative errors $\cE_{\mbox{opt}}$ and $\cE_{\mbox{equ}}$ are comparable. More importantly, for all the instances where the relative error corresponding to the equally-spaced sampling points exceeds $1\times10^{-3}$, the usage of the optimal sampling points can always bring down the relative error to below $1\times10^{-3}$. We conclude that for this example the obtained optimal sampling points are superior to the equally-spaced points.

\medskip

\noindent {\bf \large  Experiment 2: algorithm (\ref{algorithm1}), $K=$ the one-dimensional Sinc kernel, $n=8$, $\Omega=[-3,3]$, $\mu=$ the Lebesgue measure on $\Omega$.}

{\bf Figure 5.3} Distribution of the obtained 8 optimal sampling points (marked with a star) and the equally-spaced points (marked with a circle) on $\Omega=[-3,3]$.
\begin{center}
\scalebox{0.5}[0.5]{\includegraphics*{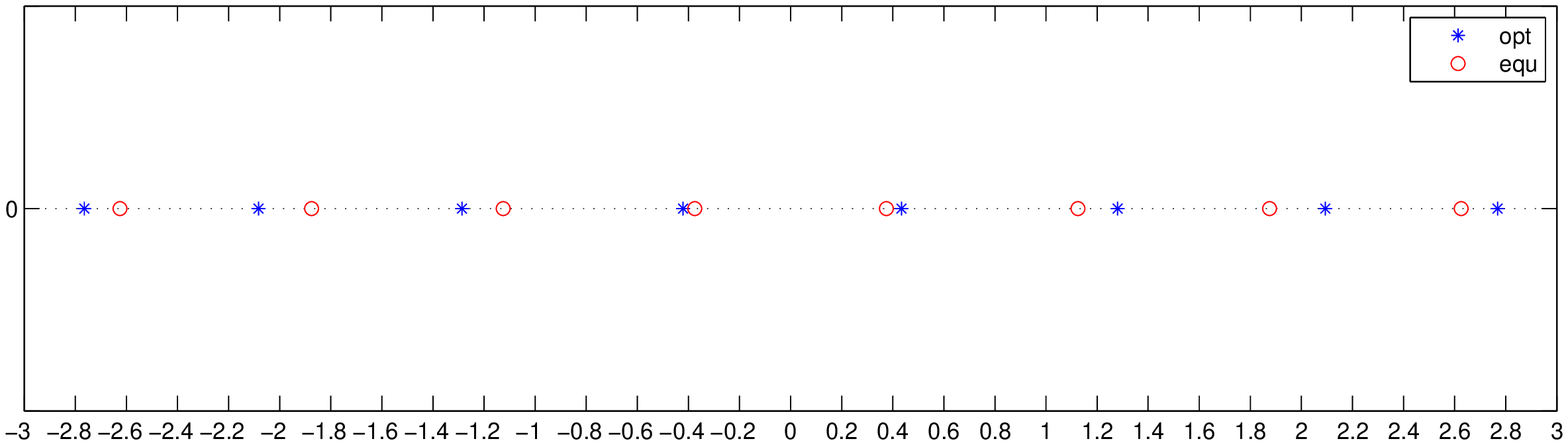}}
\end{center}

{\bf Table 5.2} The mean and standard deviation of the improvement $\cE_{\mbox{equ}}-\cE_{\mbox{opt}}$.
$$
\begin{array}{cc}
\hline\hline
\mbox{mean}&\mbox{standrad deviation}\\
0.0018 &   0.0026\\
\hline\hline
\end{array}
$$

{\bf Figure 5.4} Relative approximation errors $\cE_{\mbox{opt}}$ (marked with a circle) and $\cE_{\mbox{equ}}$ (marked with a star).
\begin{center}
\scalebox{0.5}[0.6]{\includegraphics*{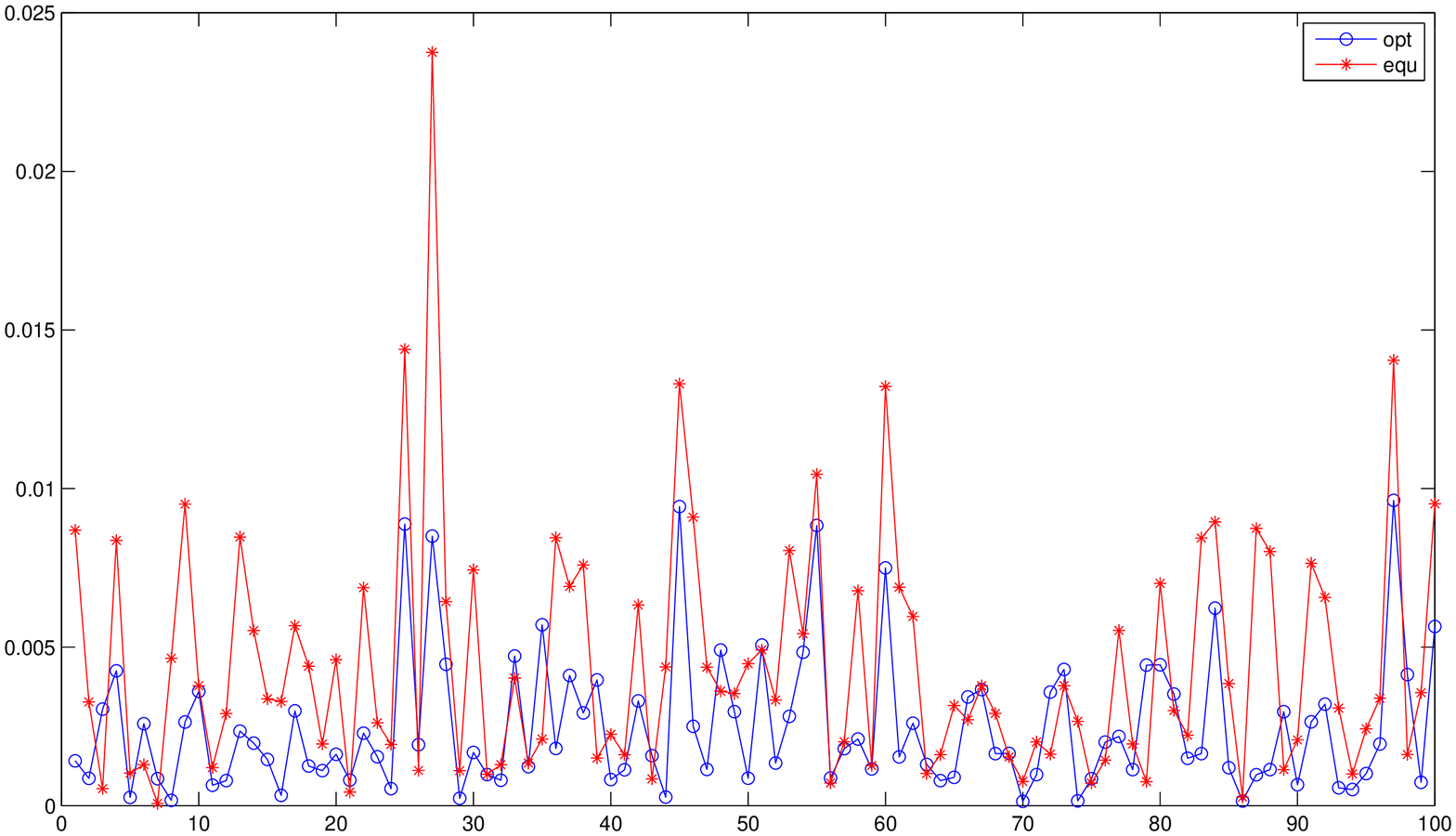}}
\end{center}

For the 100 pairs of relative approximation errors, there are 23 pairs for which $\cE_{\mbox{opt}}$ is larger than $\cE_{\mbox{equ}}$. There are 34 $\cE_{\mbox{equ}}$ (compared to 10 $\cE_{\mbox{opt}}$) that are larger than $5\times 10^{-3}$. And in 26 instances among those 34, replacing the equally-spaced points with the optimal sampling points reduces the relative approximation error to below $5\times 10^{-3}$. We also conclude that for this example the obtained optimal sampling points perform better than the equally-spaced points, although the improvement is not as drastic as Experiment 1.

\bigskip

\noindent {\bf \large  Experiment 3: algorithm (\ref{algorithm1}), $K=$ the two-dimensional Gaussian kernel, $n=36$, $\Omega=[-2,2]\times[-2,2]$, $\mu=$ the Lebesgue measure on $\Omega$.}

{\bf Figure 5.5} Distribution of the obtained 36 optimal sampling points (marked with a star) and the equally-spaced points (marked with a circle) on $[-2,2]\times[-2,2]$.
\begin{center}
\scalebox{0.5}[0.5]{\includegraphics*{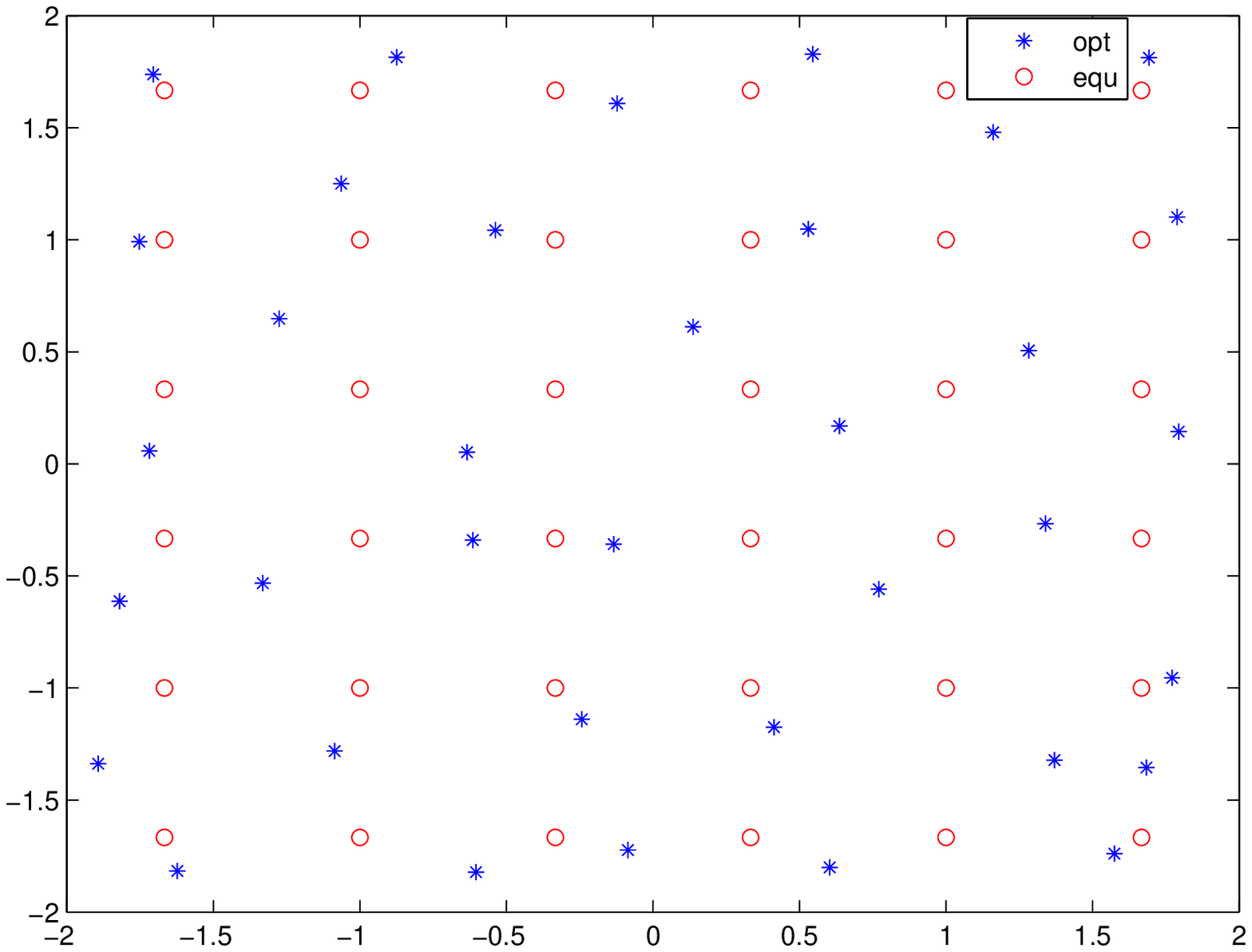}}
\end{center}

{\bf Table 5.3} The mean and standard deviation of the improvement $\cE_{\mbox{equ}}-\cE_{\mbox{opt}}$.
$$
\begin{array}{cc}
\hline\hline
\mbox{mean}&\mbox{standrad deviation}\\
0.0028  &  0.0043\\
\hline\hline
\end{array}
$$

{\bf Figure 5.6} Relative approximation errors $\cE_{\mbox{opt}}$ (marked with a circle) and $\cE_{\mbox{equ}}$ (marked with a star).
\begin{center}
\scalebox{0.8}[0.45]{\includegraphics*{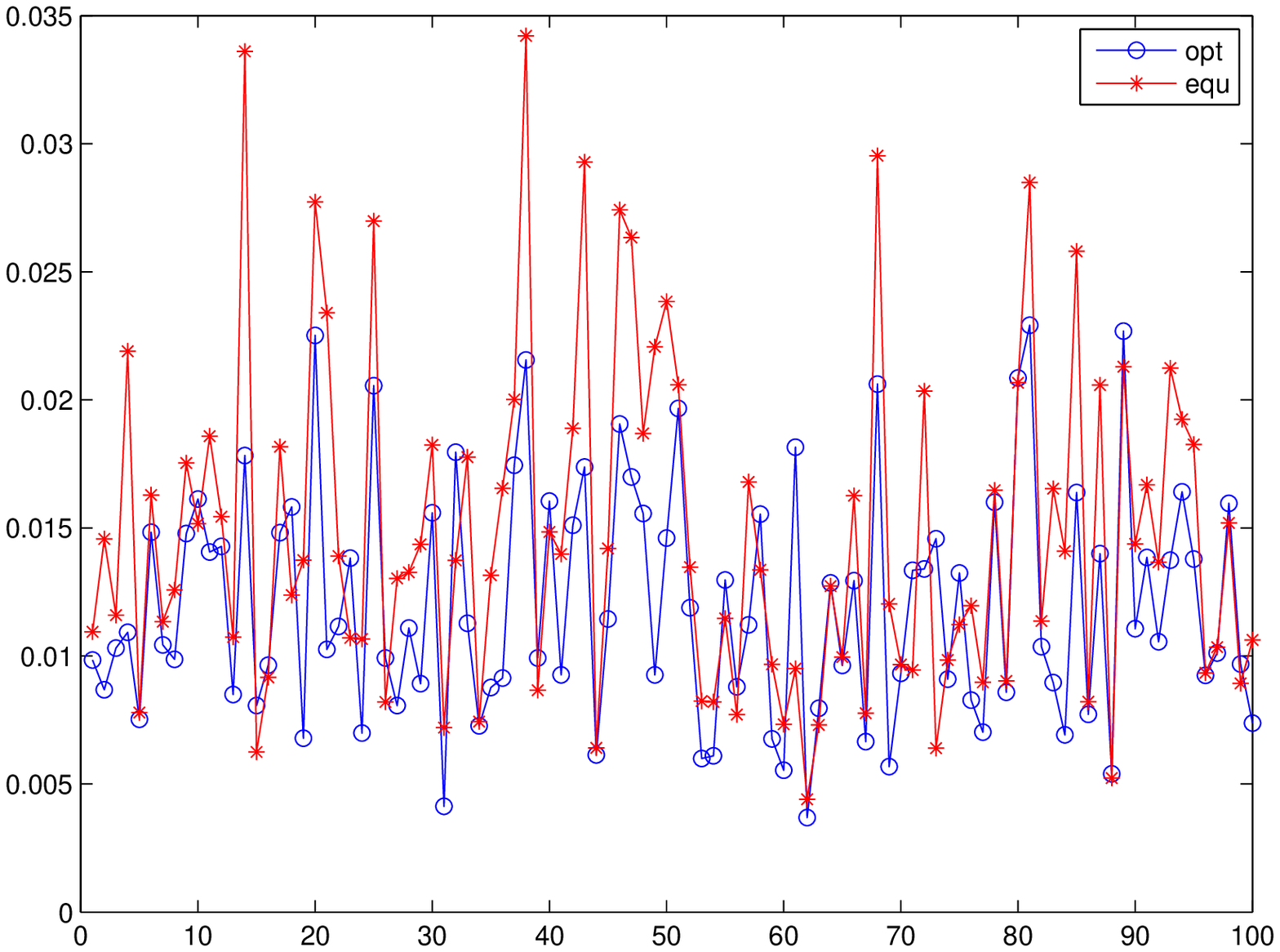}}
\end{center}

For the 100 pairs of relative approximation errors, there are 23 pairs for which $\cE_{\mbox{opt}}$ is larger than $\cE_{\mbox{equ}}$. In these pairs, $\cE_{\mbox{equ}}$ and $\cE_{\mbox{opt}}$ are rather close. We see that the value of the optimal sampling points lies in that they could dramatically reduce the relative error when the equally-spaced points perform badly. There are 10 such examples in Figure 5.6.

\bigskip

\noindent {\bf \large  Experiment 4: algorithm (\ref{algorithm1}), $K=$ the two-dimensional Sinc kernel, $n=25$, $\Omega=[-2,2]\times[-2,2]$, $\mu=$ the Lebesgue measure on $\Omega$.}

{\bf Figure 5.7} Distribution of the obtained 25 optimal sampling points (marked with a star) and the equally-spaced points (marked with a circle) on $\Omega=[-2,2]\times[-2,2]$.
\begin{center}
\scalebox{0.5}[0.5]{\includegraphics*{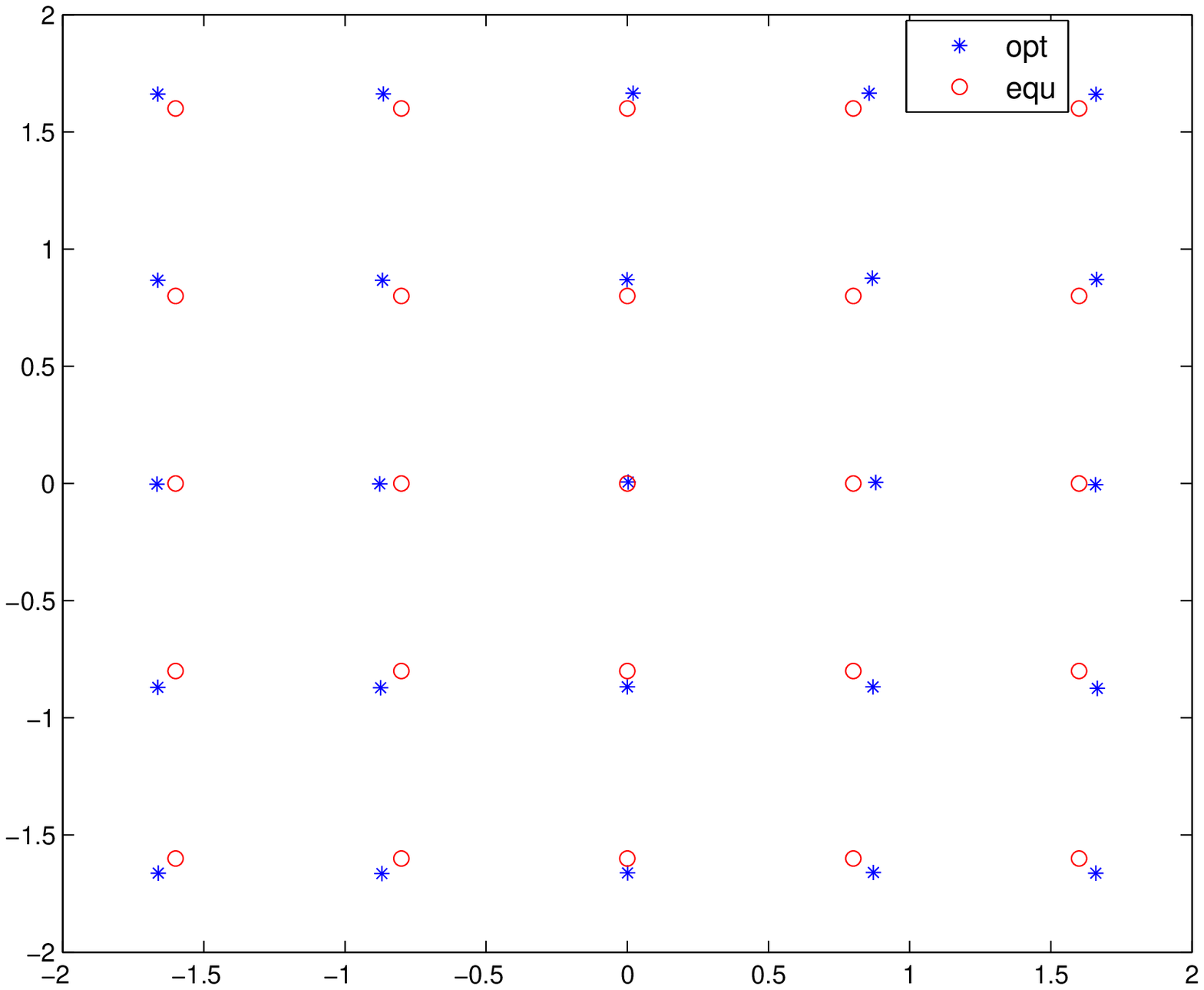}}
\end{center}

{\bf Table 5.4} The mean and standard deviation of the improvement $\cE_{\mbox{equ}}-\cE_{\mbox{opt}}$.
$$
\begin{array}{cc}
\hline\hline
\mbox{mean}&\mbox{standrad deviation}\\
0.0020  &  0.0045\\
\hline\hline
\end{array}
$$

{\bf Figure 5.8} Relative approximation errors $\cE_{\mbox{opt}}$ (marked with a circle) and $\cE_{\mbox{equ}}$ (marked with a star).
\begin{center}
\scalebox{0.7}[0.5]{\includegraphics*{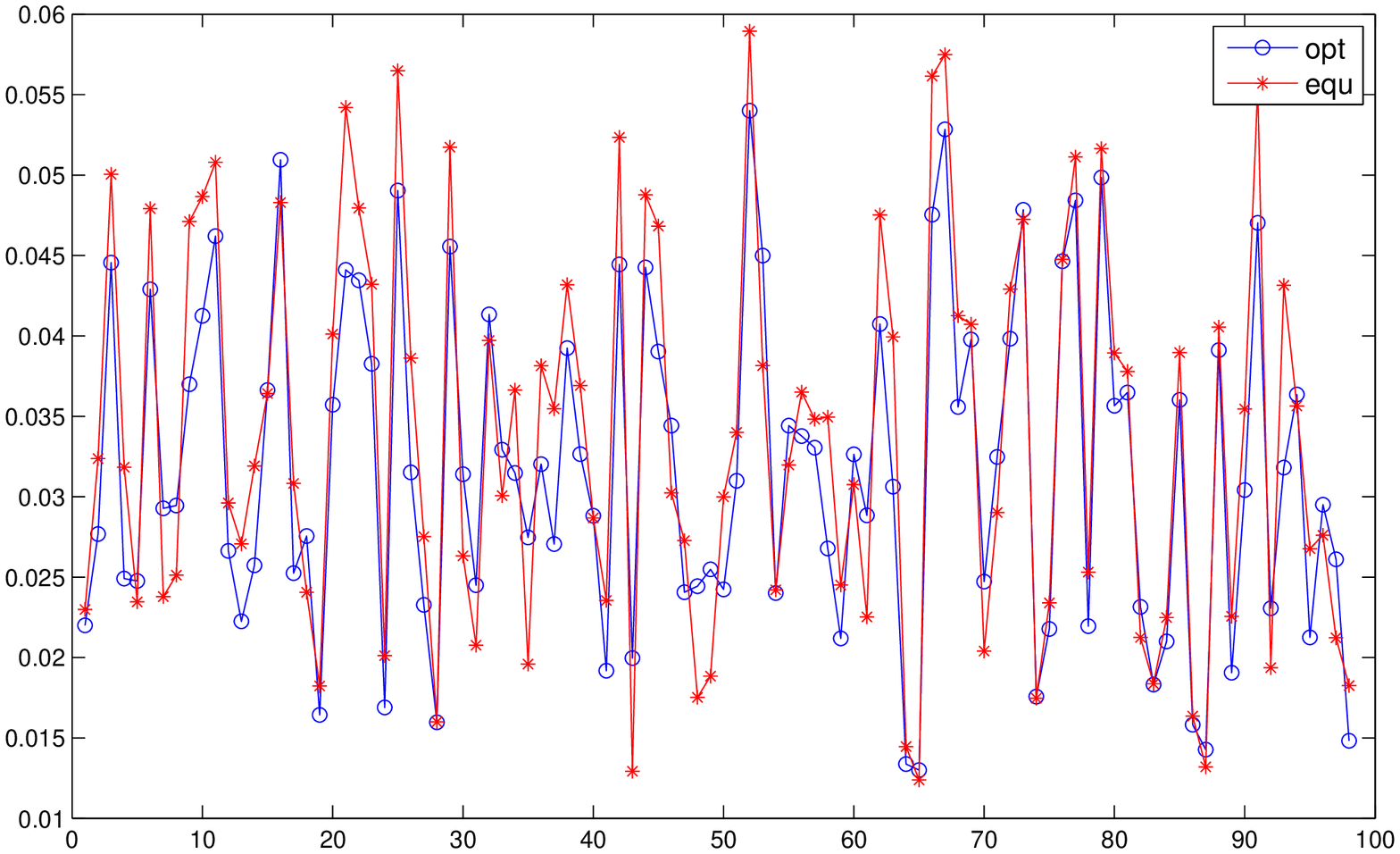}}
\end{center}

In the 100 pairs of relative approximation errors, there are 31 pairs for which $\cE_{\mbox{opt}}$ is larger than $\cE_{\mbox{equ}}$. We see from Figure 5.7 that for this example, the obtained optimal sampling points are rather close to the equally-spaced points. As a consequence, the relative approximation errors shown in Figure 5.8 are comparable.

\bigskip

In the following, we present two experiments about algorithm (\ref{algorithm3}).

\medskip

\noindent {\bf \large  Experiment 5: algorithm (\ref{algorithm3}), $K=$ the one-dimensional Gaussian kernel, $n=12$, $\Omega=[-3,3]$, $\mu$ is the uniform discrete measure supported at the 30 equally-spaced points in $\Omega$.}

{\bf Figure 5.9} Distribution of the obtained 12 optimal sampling points (marked with a star) and the equally-spaced points (marked with a circle) on $\Omega=[-3,3]$.
\begin{center}
\scalebox{0.5}[0.5]{\includegraphics*{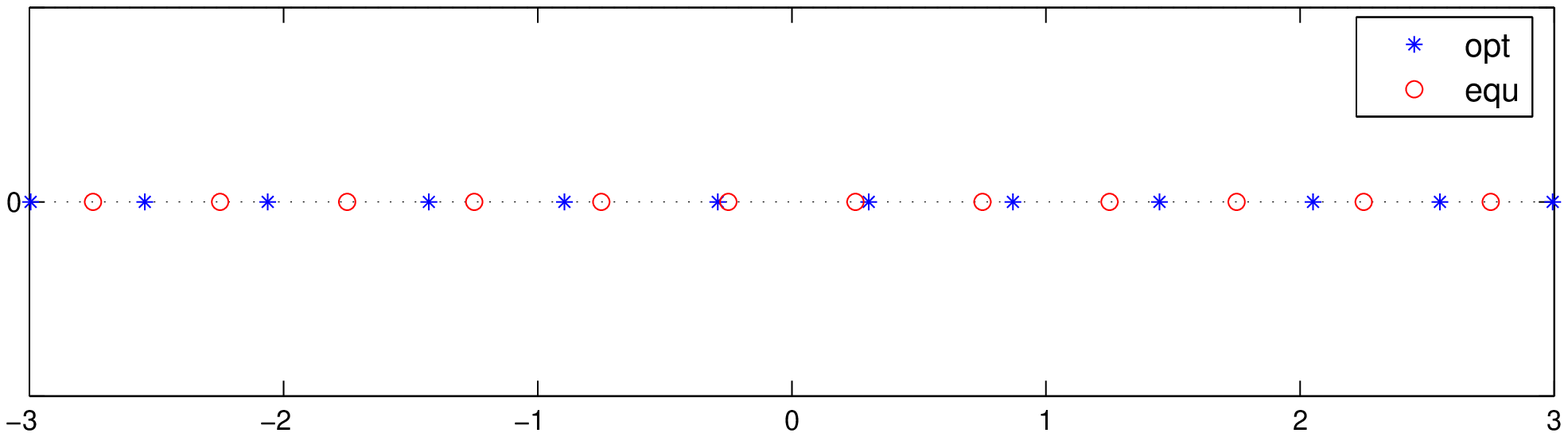}}
\end{center}

{\bf Table 5.5} The mean and standard deviation of the improvement $\cE_{\mbox{equ}}-\cE_{\mbox{opt}}$.
$$
\begin{array}{cc}
\hline\hline
\mbox{mean}&\mbox{standrad deviation}\\
  0.7528\times10^{-3}  &  0.9825\times10^{-3}\\
\hline\hline
\end{array}
$$
{\bf Figure 5.10} Relative approximation errors $\cE_{\mbox{opt}}$ (marked with a circle) and $\cE_{\mbox{equ}}$ (marked with a star).
\begin{center}
\scalebox{0.5}[0.5]{\includegraphics*{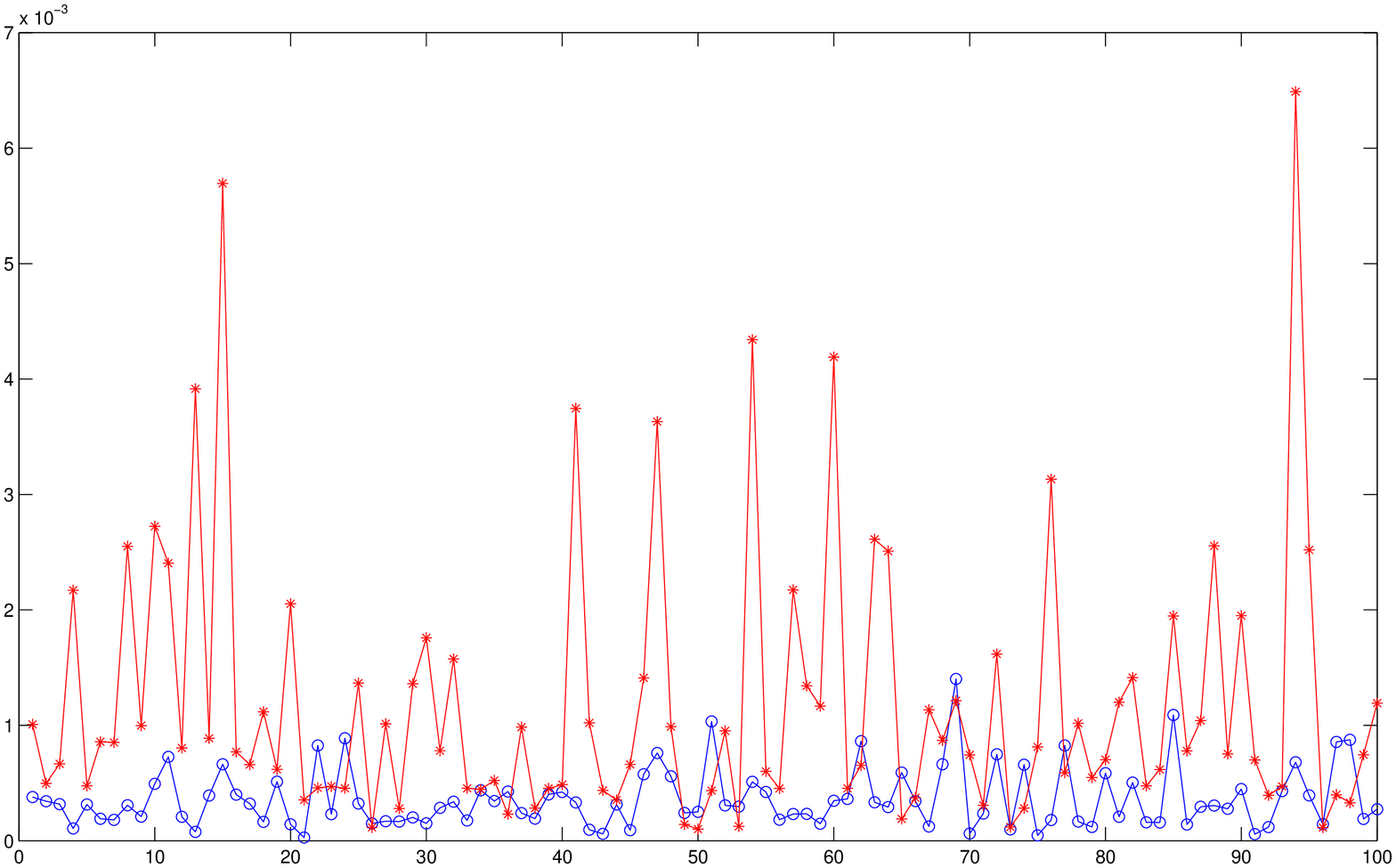}}
\end{center}

In the 100 pairs of relative approximation errors, there are only 16 pairs for which $\cE_{\mbox{opt}}$ is larger than $\cE_{\mbox{equ}}$. One sees that in those 16 instances, the relative errors $\cE_{\mbox{opt}}$ and $\cE_{\mbox{equ}}$ are comparable. For the remaining 84 instances, the improvement brought by the optimal sampling points resulting from algorithm (\ref{algorithm3}) is drastic. In particular, there are 39 instances where $\cE_{\mbox{equ}}$ exceeds $10^{-3}$ while only three $\cE_{\mbox{opt}}$ do so. Comparing results here with those in Experiment 1, one sees that algorithm (\ref{algorithm3}) is superior to (\ref{algorithm1}) for this problem.

\bigskip

\noindent {\bf \large  Experiment 6: algorithm (\ref{algorithm3}), $K=$ the one-dimensional Sinc kernel, $n=8$, $\Omega=[-3,3]$, $\mu$ is the uniform discrete measure supported at the 20 equally-spaced points in $\Omega$.}

{\bf Figure 5.11} Distribution of the obtained 8 optimal sampling points (marked with a star) and the equally-spaced points (marked with a circle) on $\Omega=[-3,3]$.
\begin{center}
\scalebox{0.5}[0.5]{\includegraphics*{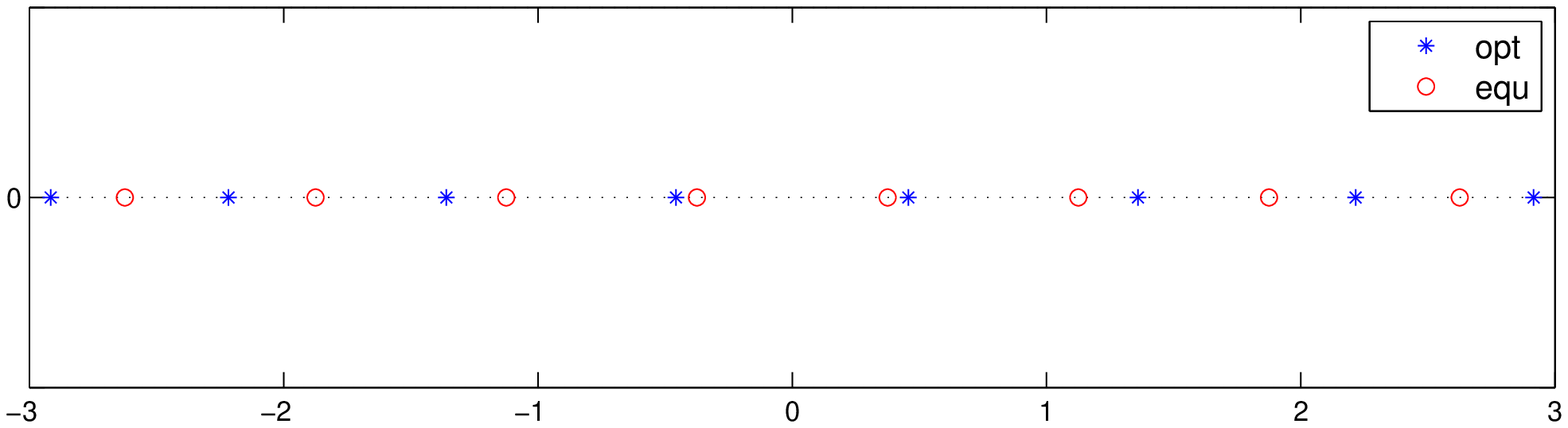}}
\end{center}

{\bf Table 5.6} The mean and standard deviation of the improvement $\cE_{\mbox{equ}}-\cE_{\mbox{opt}}$.
$$
\begin{array}{cc}
\hline\hline
\mbox{mean}&\mbox{standrad deviation}\\
  0.0035  &  0.0063\\
\hline\hline
\end{array}
$$
{\bf Figure 5.12} Relative approximation errors $\cE_{\mbox{opt}}$ (marked with a circle) and $\cE_{\mbox{equ}}$ (marked with a star).
\begin{center}
\scalebox{0.6}[0.5]{\includegraphics*{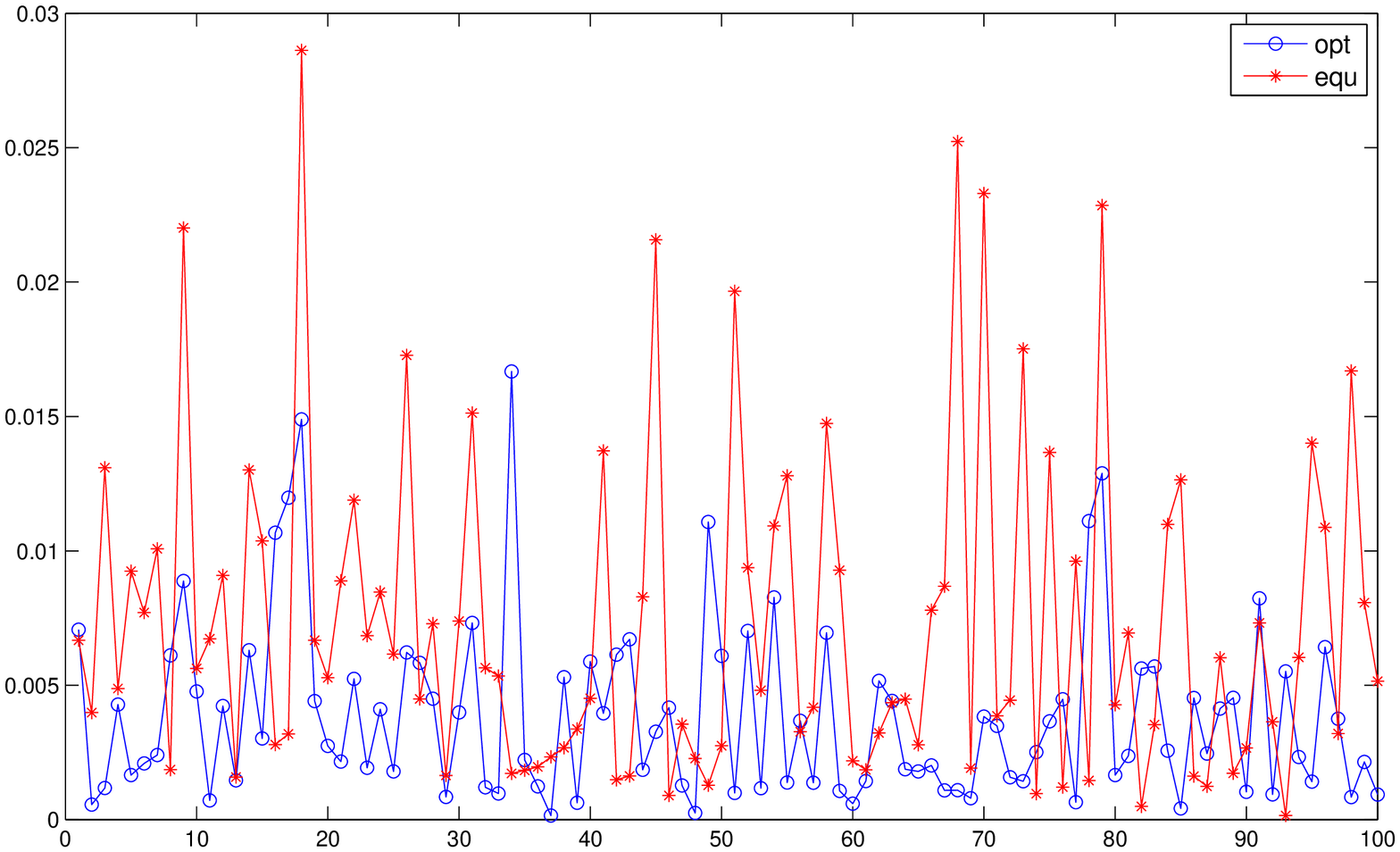}}
\end{center}

In the 100 pairs of relative approximation errors, there are only 28 pairs for which $\cE_{\mbox{opt}}$ is larger than $\cE_{\mbox{equ}}$. Except for 5 outliers, $\cE_{\mbox{opt}}/\cE_{\mbox{equ}}<5$ for those instances. For the remaining 72 improved instances, there are 21 for which $\cE_{\mbox{equ}}/\cE_{\mbox{opt}}>5$ and 8 for which $\cE_{\mbox{equ}}/\cE_{\mbox{opt}}>10$. We conclude that the optimal sampling points yielding from algorithm (\ref{algorithm3}) are significantly better than the equally-spaced points. The results here outperform those in Experiment 2.

\medskip

\section{Appendix: proof of Example \ref{example2}}
\setcounter{equation}{0}
We shall prove that the optimal sampling points for Example \ref{example2} are given by (\ref{optpoint1}) and (\ref{optpoint2}). The proof is done by considering each case of the relative location of the two sampling point with respect to the reconstruction domain $\Omega=[a,b]$.

\begin{description}
\item{Case 1}: $x_1,x_2$ lie on the right hand of $\Omega$. We set $u=x_1-b$ and $r=x_2-x_1$. Then there holds
\begin{eqnarray*}
\min_{x\in\Omega}V(x,x_1,x_2)&=&\frac{e^{-2(u+L)}+e^{-2(u+r+L)}-2e^{-2(u+r+L)}}{1-e^{-2r}}\\
&=&e^{-2(u+L)}.
\end{eqnarray*}
It is easy to see that
\begin{equation}\label{max1}
\sup_{x_1,x_2\in\mathbb{R}^d}\min_{x\in\Omega}V(x,x_1,x_2)=e^{-2L}
\end{equation}
and the supremum is achieves when $u=0$.

\item{Case 2}: $x_1,x_2$ lie on the left hand and the right hand of $\Omega$, respectively. We set $t=a-x_1$ and $s=x_2-b$. For each $x\in\Omega$, we also let $u=x-a$. By these notations, we get that
$$
\min_{x\in\Omega}V(x,x_1,x_2)=\min_{u\in[0,L]}\frac{e^{-2(u+t)}+e^{-2(L-u+s)}-2e^{-2(L+s+t)}}{1-e^{-2(L+s+t)}}.
$$
If $t\geq L+s$, we obtain that the  minimum achieves at $u=0$ and
$$
\min_{x\in\Omega}V(x,x_1,x_2)=\frac{e^{-2t}+e^{-2(L+s)}-2e^{-2(L+s+t)}}{1-e^{-2(L+s+t)}},
$$
which is decreasing with respect to $s$ and $t$. Hence, we get the conclusion that
$$
\sup_{x_1,x_2\in\mathbb{R}^d}\min_{x\in\Omega}V(x,x_1,x_2)=\frac{2e^{-2L}}{1+e^{-2L}},
$$
where the supremum achieves at $s=0$ and $t=L$. Similarly, for the case when $s\geq t+L$, we also get that
$$
\sup_{x_1,x_2\in\mathbb{R}^d}\min_{x\in\Omega}V(x,x_1,x_2)=\frac{2e^{-2L}}{1+e^{-2L}}.
$$
For the case when $|s-t|\leq L$, the minimum achieves at $u=\frac{L+s+t}{2}$ and there holds
$$
\min_{x\in\Omega}V(x,x_1,x_2)=\frac{2e^{-(L+s+t)}}{1+e^{-(L+s+t)}}.
$$
By taking the supremum of the above equation, we have
$$
\sup_{x_1,x_2\in\mathbb{R}^d}\min_{x\in\Omega}V(x,x_1,x_2)=\frac{2e^{-L}}{1+e^{-L}}.
$$
It follows from the inequality
$$
\frac{2e^{-L}}{1+e^{-L}}>\frac{2e^{-2L}}{1+e^{-2L}}
$$
that in case (2), there holds
\begin{equation}\label{max2}
\sup_{x_1,x_2\in\mathbb{R}^d}\min_{x\in\Omega}V(x,x_1,x_2)=\frac{2e^{-L}}{1+e^{-L}},
\end{equation}

\item{Case 3}: $x_0, x_1\in\Omega$. We set $u=x_0-a, v=b-x_1$ and $r=x_1-x_0$. If $x=a$, we have that
$$
V(x,x_1,x_2)=\frac{e^{-2u}+e^{-2(u+r)}-2e^{-2(u+r)}}{1-e^{-2r}}=e^{-2u}.
$$
Similarly, we also get for $x=b$ that
$$
V(x,x_1,x_2)=e^{-2v}.
$$
If $x\in[x_1,x_2]$, there holds
$$
V(x,x_1,x_2)=\frac{e^{-2|x-x_0|}+e^{-2(r-|x-x_0|)}-2e^{-2r}}{1-e^{-2r}}.
$$
Thus the minimum of $V$ achieves at $|x-x_0|=\frac{r}{2}$ and there holds
$$
\min_{x\in[x_1,x_2]}V(x,x_1,x_2)=\frac{2e^{-r}}{1+e^{-r}}.
$$
According to the above discussion, we need to consider
$$
\sup_{u\geq v}\min\left\{e^{-2u}, e^{-2v}, \frac{2e^{-(L-u-v)}}{1+e^{-(L-u-v)}}\right\}.
$$
It is not difficult to see that
$$
\sup_{u\geq v}\min\left\{e^{-2u}, e^{-2v}, \frac{2e^{-(L-u-v)}}{1+e^{-(L-u-v)}}\right\}=e^{-2u},
$$
where there holds
\begin{equation}\label{equationu}
e^{-2u}=\frac{2e^{-(L-2u)}}{1+e^{-(L-2u)}}.
\end{equation}
By solving equation (\ref{equationu}), we obtain
$$
u=-\frac{1}{2}\ln\left(\frac{-e^{-L}+\sqrt{e^{-2L}+8e^{-L}}}{2}\right)
$$
and
\begin{equation}\label{max3}
\sup_{u\geq v}\min\left\{e^{-2u}, e^{-2v}, \frac{2e^{-(L-u-v)}}{1+e^{-(L-u-v)}}\right\}=
\frac{-e^{-L}+\sqrt{e^{-2L}+8e^{-L}}}{2}.
\end{equation}

\end{description}

It remains to compare (\ref{max1}), (\ref{max2}) and (\ref{max3}). By calculation, we have
$$
\frac{-e^{-L}+\sqrt{e^{-2L}+8e^{-L}}}{2}>\frac{2e^{-L}}{1+e^{-L}}>e^{-2L},
$$
which implies the optimal two sampling points should be placed inside $\Omega$ by (\ref{optpoint1}) and (\ref{optpoint2}).
\end{document}